\documentclass{amsart}
\usepackage{amsmath}
\usepackage{amssymb}
\usepackage{amscd}
\usepackage{amsthm}
\usepackage{amsfonts}
\usepackage{amsxtra}
\usepackage[all]{xy}
\SelectTips{cm}{}
\usepackage[numbers,sort&compress]{natbib}
\usepackage{array}
\usepackage{booktabs}
\usepackage{mathtools}
\usepackage{enumitem}
\usepackage{etoolbox}
\usepackage{xcolor}
\usepackage{graphicx}
\usepackage{subfig}

\usepackage{rotating}
\usepackage[unicode=true,bookmarks=false,backref=false]{hyperref} 
\hypersetup{breaklinks=true,pdfborder={0 0 0},pdfborderstyle={},colorlinks=true}
\hypersetup{pdfencoding=auto, linkcolor=blue, citecolor=[rgb]{0.0,0.7,0.1}, linktoc=all}

\preto{\section}{\vskip 1in}

\setlist[description]{leftmargin=0pt, itemsep=\bigskipamount, topsep=\medskipamount, labelsep*=3pt}

\newcommand{\emptylabel}{\hspace{-3pt}} 

\makeatletter
\@namedef{subjclassname@2020}{%
    \textup{2020} Mathematics Subject Classification}
\makeatother

\newcolumntype{E}[1]{>{\raggedleft\arraybackslash}m{#1}@{\extracolsep{0pt}\;=\;}}
\newcolumntype{F}[1]{>{\raggedright\arraybackslash}m{#1}}

\newcolumntype{M}[1]{>{\bfseries}m{#1}}

\theoremstyle{plain}

\newtheorem{Thm}{Theorem}[section]

\newtheorem{Pro}[Thm]{Proposition}
\newtheorem{Cor}[Thm]{Corollary}

\theoremstyle{definition}
\newtheorem{Def}[Thm]{Definition}

\newtheorem{Rem}[Thm]{Remark}

\newtheorem*{Rem-intro}{Remark}

\newcommand{\ep}{\varepsilon}

\DeclareMathOperator{\Hom}{Hom}

\newcommand{\Dirlim}{\varinjlim}
\newcommand{\Invlim}{\varprojlim}

\DeclareMathOperator{\Tr}{Tr}

\DeclareMathOperator{\inde}{index}
\DeclareMathOperator{\rk}{rk}
\DeclareMathOperator{\sgn}{sign}

\newcommand{\ZZ}{{\mathbb{Z}}}
\newcommand{\QQ}{{\mathbb{Q}}} 

\newcommand{\CC}{{\mathbb{C}}}

\newcommand{\NN}{{\mathbb{N}}}

\newcommand{\RR}{{\mathbb{R}}}
\newcommand{\PP}{{\mathbb{P}}}

\newcommand{\WW}{{\mathcal{W}}}

\newcommand{\SSS}{{\mathcal{S}}}

\newcommand{\FF}{{\mathbb{F}}}
\newcommand{\cR}{{\mathcal{R}}}
\newcommand{\cF}{{\mathcal{F}}}

\newcommand{\lp}{\textup{(}}
\newcommand{\rp}{\textup{)}}
\newcommand{\co}{\colon\,}
 
\newcommand{\hull}{{\Omega _T}}
\newcommand{\Cech}{\v{C}ech} 
\newcommand{\CH}{\mathbf{\check H}}
\newcommand{\Nc}{\mathcal{N}}
\newcommand{\tk}{{\tau^K_*}}
\newcommand{\tc}{{\tau^{\check H}_*}} 
\newcommand{\zone}{{\mathcal{B}}}

\newcommand{\ch}{ch}              
\newcommand{\id}{id}              

\begin{document}

\title[Diffraction and Spectral Data of Aperiodic Tilings]
{Relating Diffraction and Spectral Data of Aperiodic Tilings:  
Towards a Bloch theorem}

\author[Akkermans]{Eric Akkermans}
\address{Department of Physics,
    Technion, 
    Haifa 3200003, Israel}
\email{eric@physics.technion.ac.il}

\author[Don]{Yaroslav Don}
\address{Department of Physics,
    Technion, 
    Haifa 3200003, Israel}
\email{yarosd@campus.technion.ac.il}

\author[Rosenberg]{Jonathan Rosenberg}
\address{Department of Mathematics,
    University of Maryland, 
    College Park, MD 20742, USA}
\email{jmr@math.umd.edu}

\author[Schochet]{Claude L.~Schochet}
\address{Department of Mathematics,
    Technion,
    Haifa 3200003, Israel}
\email{clsmath@gmail.com}

\thanks{JR partially supported by {U.S.} NSF grant number DMS-1607162. This work was supported by the Israel Science Foundation Grant No.~924/09 and by the Pazy Research Foundation.}

\begin{abstract}

The purpose of this paper is to show the relationship in all dimensions between the structural (diffraction pattern) aspect of tilings
(described by \v Cech cohomology of the tiling space) and  the spectral properties
(of Hamiltonians defined on such tilings) defined by $K$-theory, and to show 
their equivalence in dimensions $\leq$ 3. 
A theorem makes precise the conditions for this relationship to hold.  It can be viewed 
as an extension of  the \lq\lq Bloch Theorem\rq\rq\,
to a large class of aperiodic tilings. 
The idea underlying this result is based on the     relationship between cohomology and  $K$-theory traces
and their equivalence in low dimensions.  

\end{abstract}

\maketitle

\tableofcontents


\section {Introduction}

Aperiodic tilings are  structures obtained from the spatial arrangement of letters defining an alphabet, according to a set of deterministic rules~\citep{Senechal_Book_1996,Janot_Book_1995,Baake_Grimm_Book_2013,Baake_Grimm_Book_2018}. They constitute a rich playground to investigate features of physical systems in different contexts, e.g.\ condensed matter, statistical mechanics and dynamical systems. 

This ubiquity is partly due to the existence of a large set of tiling families which includes periodic, nonperiodic (e.g.\ Wang tiles), quasiperiodic, and asymptotically periodic tilings. For periodic tilings, the Bloch theorem~\citep{Ashcroft_Mermin_Book_1976} provides a systematic and powerful relation between different aspects, such as diffraction and spectral data. For aperiodic tilings  such as quasicrystals, despite having been  thoroughly studied, these aspects  remain as yet unrelated since the classical Bloch theorem is not applicable.

A celebrated family of tilings are quasicrystals or quasiperiodic tilings related to algebraic number theory and cut-and-project (C\&P) sets~\citep{DuneauKatz_1985,KatzDuneau_1986}. Despite their lack of periodicity, quasicrystals discovered by Shechtman \citep{Shechtman_1984}, and predicted by Levine and Steinhardt \citep{Levine_1984}, exhibit sharp Bragg peaks.  Quasicrystals have been extensively investigated \citep{Luck_Review_1994,Senechal_Book_1996}, especially in one dimension  \citep{BT}. 

Quasicrystals have also been studied from the viewpoint of the spectral characteristics of the waves (acoustic, optical, matter) they can sustain. Conveniently defined Laplacians (continuous or tight-binding) reveal a highly lacunar fractal energy spectrum, with an infinite set of energy gaps~\citep{Luck_PRB_1989,Damanik_CMP_2008,Tanese_2014,Sutherland_PRB_1986}.

Johnson and Moser \citep{Johnson_1982,Johnson_1986}  studied  the spectrum of self-adjoint linear differential operators, e.g.\ continuous Schr\"odinger or Helmholtz equations, with a potential being an almost periodic function,   and presented a systematic way of enumerating the open intervals of the associated resolvent operator  using the rotation number. This was an approach to gap labeling in the spirit of the Schwartzman winding number \cite{Schwartzman} and using cohomology ideas. A discrete version is in \citep{delyon1983rotation}. This description was based on the use of the rotation number, a quantity which from Sturm-Liouville theory equals half the counting function. The theorem of Johnson and Moser applies to one-dimensional systems with a quasiperiodic potential. 
The Gap Labelling Theorem (hereafter GLT)~\citep{Bellissard_RMP_1991,Bellissard_Review_1992} of Bellissard and coworkers provides a more general framework for the topological classification of these gaps and plays, for quasiperiodic tilings, a role similar to that of the Bloch theorem for periodic ones. The Bloch theorem makes it possible to label the eigenstates of a periodic system with a quasi-momentum and to identify topological (Chern) numbers ~\citep{TKNN_1982}. 
This labelling is robust as long as the lattice translation symmetry is preserved. Similarly, the GLT allows one to associate numbers to each gap in the spectrum of quasiperiodic tilings. Those numbers can be given both a topological meaning and invariance properties akin in nature to Chern numbers, but not expressible in terms of a [classical] curvature~\citep{Bellissard_RMP_1991,Kunz_1986}.%
\footnote{
    The Chern
    numbers \emph{are} related to curvature in the noncommutative geometry
    sense of Connes \cite{C-CR,C}, however.
}
In both cases, topological invariants attached to the energy spectrum remain unchanged under a perturbation of the Hamiltonian, as long as gaps do not close.
Fractal features often show up in the diffraction patterns of aperiodic tilings  and in the spectral properties of related Laplacians, and have been suggested as a kind of generalization of the Bloch theorem for tilings~\citep{Luck_PRB_1989}. A relation between the spectrum of dynamical systems and Bragg peaks at the basis of the GLT has also been advocated in \citep{hof1995,Dworkin,Q,itzortiz2004eigenvalues}.

Current lore of topology of periodic and aperiodic tilings emphasizes the existence of different, even incompatible, classes of topological invariants. 
The band structure of periodic tilings, predicted by the Bloch theorem, naturally introduces an  inherent torus topology of the Brillouin zone. Aperiodic tilings do not enjoy these benefits of a Bloch theorem. Nonetheless, as obtained from the GLT, their energy spectra present a ramified Cantor set gap structure, e.g.\ for 1D quasiperiodic C$\&$P tilings, an infinite set of gaps which can be labeled using two integers~\citep{Bellissard_RMP_1991}. The topological $K$-theoretical nature of these integers has been emphasized~\citep{Bellissard_Review_1992}. Nontrivial topology for quasiperiodic tilings has also been reported, and the resulting topological invariants have been related to the gap labeling integers and winding numbers.
This is connected to scattering data and diffraction spectra of aperiodic tilings~\citep{Dareau_PRL_2017,Baboux_2017}. The relevance of \v Cech cohomology as an important tool in the study of substitution tilings has been emphasized by Anderson and Putnam \citep{AP} and also in \citep{AR}. Given such a tiling space, they show that it is topologically conjugate
to an inverse limit of explicit finite complexes, and hence 
the cohomology is readily computable. Anderson-Putnam remark: \lq\lq \emph{Our point of view is that
it is \v Cech cohomology which is really measuring the 
almost periodic structure of these tilings.}" 

Both characterisations of tilings, diffraction spectra (Bragg peaks) and spectral data for wave equations, are obtained by means of conveniently defined traces expressed either by a two-point correlation function or by the integrated density of states (a.k.a.\ counting function). At this point, there are two communities of mathematicians and physicists who build 
these traces from Ruelle-Sullivan currents.
One group  uses \v Cech cohomology 
and defines the cohomology trace $\tc$ there \citep{AP,AR}. The other group uses $K$-theory
almost exclusively and defines the $K$-theory trace $\tk$  there \citep{BBG,shu,LPV,M,Kr,BM}.  The goal of this paper is to unite the two groups of people, to show that for a large class of tilings, including
cut-and-project (C$\&$P) aperiodic tilings, these two traces are equivalent, at least in dimension $\le 3$, a result which can be seen as an extension of the Bloch theorem to this family of aperiodic tilings. To that purpose, we will first
show that the two approaches give exactly the same trace under very 
general circumstances and then we will show by way of examples how
results from both groups fit under the same umbrella.  In dimensions
$4$ and up, part of the GLT still seems to be in doubt since there
is no convincing proof of integrality of the Chern character in
the literature.  But the equivalence between cohomology and $K$-theoretic
traces is still valid up to perhaps an integral factor.  We will
explain this in Section \ref{sec:uniting}.

This effort requires quite a bit of background knowledge. In order to
keep the paper to a reasonable length, we will specify the tools 
that we need from foliated spaces, measure theory, cohomology and
$K$-theory, and we will give precise references in the literature to
theorems that we require. 

In order to demonstrate the generality of our results, we will wait to
introduce tiling terminology and conditions until needed. For now, 
we assume given a compact foliated space $X$ with oriented foliation bundle. 
Our basic reference for foliated spaces is the book of Moore and Schochet
\cite{MS}.

A \emph{foliated chart} or \emph{foliated patch} in a
topological space is an open set
homeomorphic to $L \times N$, where $L$ is a copy of
$\RR ^d$ and $N$ is a separable locally compact
metrizable space. A \emph{tangentially smooth} function
$f \co L \times N \to \RR $ is a continuous function such that
$f(\,\bullet\, , n)$ is smooth for each 
$n \in N$ and the partial derivatives of $f$ in the $L$ direction are
continuous on $L \times N$.  A \emph{foliated space}  \cite[p.~32]{MS}
$X$ is a separable locally compact metrizable space with an open covering
by foliated charts fitting together smoothly so that the local plaques
$L \times \{n\}$ fit together to form $d$-dimensional 
smooth manifolds called \emph{leaves}.  Foliated manifolds are the classical
examples of foliated spaces, but for us the relevant examples are tiling
spaces, which are foliated spaces but typically are not foliated manifolds.  We let $C_\tau ^\infty (X)$
denote the tangentially smooth functions $f\co X \to \RR$; that is, they are
tangentially smooth when restricted to every local patch. 
A foliated space has a natural $d$-dimensional real tangent bundle $F$ along the
leaves.  Its dual bundle is denoted $F^*$. 

The foliated spaces relevant to tiling theory are quite special.
In order to take advantage of that fact, we shall specialize at once. 

\begin{Def}
\label{def:foliatedRd}
A \emph{compact foliated space given by an $\RR^d$-action}
is a compact foliated space $X$ with a locally free $\RR ^d$-action, 
such that the orbits of the action are the leaves of the foliated
space, and
$X \cong  N \times_{\Lambda}\RR ^d$ for some compact totally disconnected
space $N$ carrying an action of a lattice $\Lambda\subset \RR ^d$.
There is a resulting fibre bundle
\[
N \longrightarrow X \xrightarrow{\ p\ } T^d,
\]
where $T^d=\RR^d/\Lambda$ is the $d$-torus, 
and the restriction of the projection to each leaf $\ell$
is a covering map $\ell \to T^d$. 
\end{Def}

Henceforth, \textbf{all foliated spaces and in particular all tiling
spaces will be assumed to be compact foliated spaces given by
an $\RR ^d$-action}.
This implies that the foliation tangent bundle is orientable. 
Sadun and Williams \cite{SW} show that the hulls of most tilings
are homeomorphic to spaces satisfying these conditions.  In general,
the homeomorphism   $X \cong  N \times_{\Lambda}\RR ^d$   is \emph{not}
equivariant for the $\RR^d$-action on the
tiling hull, but this won't matter since all we need is the
foliated space structure, not the group action, and the homeomorphism
sends leaves to leaves. We comment on this result in detail below.

\medskip

The remaining sections are organized as follows.

\begin{description}

\item[Section 2] is a very quick introduction to \v Cech cohomology
$\CH ^*(X ; \FF ) $ for compact spaces. 

\item[Section 3] introduces tangential cohomology $H_\tau ^*(X)$ and
homology, sculpted for foliated spaces. There is a natural map 
\[
s\co\CH ^k(X ; \RR ) \longrightarrow H_\tau ^k(X)
\]
which in our special context is shown by a spectral
sequence comparison argument to be induced by
the inclusion $C^\infty(N)\hookrightarrow C(N)$.

\item[Section 4] is devoted to coinvariants. We demonstrate that if 
$X = N \times _\Lambda \RR ^d $ is a compact foliated space given by an
$\RR ^d$-action, then there is a natural isomorphism
\[
\CH ^0(N ; \RR)_{\ZZ ^d} \cong  \CH ^d(X ; \RR )
\]
and the map $s\co \CH ^d(X ; \RR ) \to H_\tau^d(X)$ has dense image in the Hausdorff quotient $\bar H_\tau^d(X)$.

\item[\textmd{In} Section 5] we introduce the machinery of topological groupoids,
tangential measures, invariant transverse measures $\nu$, and 
Ruelle-Sullivan currents $C_{\nu}$. We highlight the Riesz representation
theorem, which identifies the group of signed Radon invariant transverse 
measures with the top tangential homology group. This allows us to define the cohomology trace
\[
\tc\co \CH^d(X ;\RR ) \longrightarrow \RR
\]
as the composition
\[
\CH^d(X ;\RR ) \xrightarrow{s}   H_\tau ^d(X) 
\xrightarrow{\cap \,{C_\nu}}   \RR
\]

\item[Section 6] contains a very brief introduction to topological
$K$-theory for $C^*$-algebras and the classical Chern character. 
Following Bellissard, we define the \emph{non-commutative Brillouin zone} to be 
$\zone = C^*(G(X))$, where $G(X)$ is the holonomy groupoid of the foliated space. 
Using Connes' Thom isomorphism theorem
\[
\varphi\co K_d(A)  \xrightarrow{\cong} K_0(A \rtimes \RR^d)
\]
we record an isomorphism 
\[
K^d(X) \xrightarrow{\varphi}
K_0(C(X) \rtimes \RR^d) \cong{K_0(C^*(G(X))} = K_0(\zone )
\]
which holds in our context; we denote it $\chi$.

\item[Section 7] gives the analytical background for the $K$-theory trace.

\item[\textmd{The subject of} Section 8] is the partial Chern character,
which is a map
\[
c: K_0(C^*(G(X)) \longrightarrow \bar H_\tau ^d(X). 
\]
defined for foliated spaces.  Its existence depends upon the identification
of the invariant transverse measures with homology classes via the
Riesz theorem.  Given an invariant
transverse measure $\nu$, the $K$-theory trace, which we denote $\tk$,  is simply the composition
\[
K_0(\zone) =   K_0(C^*(G(X))) \xrightarrow{c}
\bar H_\tau^d (X) \xrightarrow{\cap C_\nu} \RR. 
\]

\item[Section 9] contains our most general result, Theorem \ref{thm:main},
relating the cohomology and $K$-theory traces.  

\theoremstyle{plain}
\newtheorem*{thm:main}{Theorem \ref{thm:main}}
\begin{thm:main}
Suppose that $X$ is a compact foliated space given by an
$\RR ^d$-action with invariant transverse measure $\nu $,
and the holonomy cover of each leaf is simply connected. Then:
\begin{itemize}
\item  The diagram 
\begin{equation}
\tag{\ref*{eq:main1}}  
\xymatrix{
K_0(\zone)\ar@2{-}[r]& K_0(C^*(G(X))) \ar[r]^(.6){\chi ^{-1}}_(.6)\cong \ar[d]^c
& K^{-d}(X) \ar[r]^{\ch_d} & \CH^d(X; \RR )\ar[d]^s\\
&H_\tau^d(X)\ar[rr]^{\\id}&&H_\tau^d(X)
}
\end{equation}
commutes. 
\medskip

\item Bloch Theorem: For every invariant transverse measure $\nu$, the diagram 
\begin{equation}
\tag{\ref*{eq:main2}}  
\xymatrix@C+6ex{
K_0(\zone)\ar[r]^(.45){(\ch_d)\circ(\chi^{-1})} \ar[d]^{\tk}
& \CH^d(X; \RR)\ar[d]^{\tc}\\
\RR \ar[r]^{\id} & \RR
}
\end{equation}
commutes. 
\end{itemize}
\end{thm:main}

We then have the following consequence for tilings.  
Suppose that $T$ is a tiling satisfying the following conditions:
\medskip
\begin{enumerate}
\item $T$ satisfies the finite pattern condition (i.e.,
finite local complexity);
\item $T$ has only finitely many tile orientations.
\end{enumerate}

\newtheorem*{cor:maintiling}{Corollary \ref{cor:maintiling}}
\begin{cor:maintiling}
Under the assumptions above, if the dimension is $\le 3$, then for every invariant transverse measure the diagram 
\begin{equation}
\tag{\ref*{eq:maintiling}}  
\xymatrix{
K_0(\zone _T)  \ar@{>>}[rr]^(.45){(\ch _d) \circ (\chi ^{-1})}  \ar[d]^{\tk}
&&\CH^d(\hull ; \ZZ ) \ar[d]^{\tc} \\
\RR \ar[rr]^{\id} &&\RR
}
\end{equation}
commutes. In particular,  the $K$-theory trace   
$\tk\co K_0(\zone_T) \to \RR$  and the cohomology trace  $\tc\co\CH^d(\hull;\ZZ) \to \RR$  have the same image in $\RR$.
\end{cor:maintiling}

\item[Sections 10] gives a quick introduction to the building tools underlying basic families of tilings.

\item[Sections 11 and 12] present a comparison of the diffraction spectrum and the counting function (gap labeling) for three canonical families of one-dimensional tilings: periodic, quasiperiodic, and aperiodic. For the quasiperiodic and aperiodic cases, the diffraction spectrum involves three possible classes: pure-point (i.e., a discrete and countable set of Bragg peaks), absolutely continuous or singular continuous. The corresponding spectrum of Laplacians on these tilings involves also these three classes although the two spectra do not necessarily coincide. For the periodic case, Laplacian and diffraction spectra are pure-point and the Bloch theorem gives the equivalence between those two data sets.

\item[Section 13] gives some insights to a selection of earlier works closely related to ours. We emphasize the distinctions between the different definitions and results.

\end{description}

\medskip

The authors wish to thank Jerry Kaminker, Ian Putnam, Ken Brown,
and Benji Weiss	for their assistance with various points in the paper and 
to acknowledge our debt to Alain Connes, who first proved 
an index theorem for foliations.


\section{\v Cech Cohomology}

The classical reference for {\Cech} cohomology for compact (Hausdorff)
spaces is Eilenberg-Steenrod \cite{ES}. Suppose that $X$ is
a  compact space. 
Then $\CH ^j(X ; \FF) $ is defined for $\FF $ any commutative ring.
(This is the same as the sheaf cohomology of the constant sheaf defined
by $\FF$.)
If $X = \Invlim X_j $ is the inverse limit of finite $CW$-complexes (and 
every compact space may be written in such a manner), then the
natural maps $H^*(X_j ; \FF) \to \CH^*(X; \FF)$ induce an
isomorphism for each $k$:
\[
\Dirlim H^k(X_j ; \FF )   \xrightarrow{\cong}  \CH^k(X; \FF )  .
\]
Note that we do not have to specify the type of cohomology
(\Cech, singular, simplicial, $\dots$) that we use for the
finite complexes, since they all agree for such spaces.
If $X$ is also separable metrizable then limits can be taken over
sequences. It follows that if $X$ is compact separable metrizable,  then 
$\CH^k(X; \ZZ )$ is the direct limit of a sequence of finitely generated
abelian groups, hence countable, and $\CH^k(X ; \RR )$ is the direct 
limit of a sequence of finite-dimensional vector spaces, hence a vector space
of (at most) countable dimension.  In addition, the natural map
$\CH^k(X ; \ZZ) \longrightarrow \CH^k(X ; \RR)$ induces an isomorphism 
\[
\CH^k(X ; \ZZ)\otimes \RR \xrightarrow{\cong} \CH^k(X ; \RR).
\]
If $X$  is compact metric of dimension $d$, then $\CH^k (X; \FF)$ is
defined and vanishes for $k > d$. 
The group $\CH ^d (X; \FF)$ is of special interest, and we will
discuss it below.      
The transversal $N$ in our applications is a Cantor set,
zero-dimensional, so its cohomology vanishes in positive dimensions.
Note that $N$ is the inverse limit of a sequence of finite
discrete spaces $X_n$, $N = \Invlim X_n$.  Thus
$\CH ^0(N ; \FF ) \cong \Dirlim H^0(X_n; \FF )$.  Each
$H^0(X_n; \FF )$ is  a finitely generated free $\FF$-module and the
maps in the direct system all split, so
$\CH ^0(N ; \FF ) \cong \oplus \FF $, where the sum is over countably many
copies of $\FF $.


\section{Tangential Cohomology and Homology}

A convenient reference for tangential cohomology is \cite[Ch.\ III]{MS}.
Tangential cohomology is referred to by various other names in the
tiling literature.%
\footnote{
    For example, 
    Kellendonk-Putnam \cite[p.\ 695]{KP} page call it \emph{dynamical
    cohomology} and generalize it. Moustafa \cite{M} calls it
    \emph{longitudinal cohomology.}
}
It is defined on   
foliated spaces $X$ (say, of leaf dimension $d$).
Recall that  $C_\tau ^\infty(X)$ is the 
space of  real-valued continuous functions on $X$ 
that are smooth in the leaf directions.  
Let $\Gamma _\tau (F^*)$ denote the tangentially smooth sections of $F^*$, the dual
of the tangent bundle to the leaves, and let
$\Omega _\tau ^k(X) = \Gamma _\tau (\bigwedge ^kF^*)$ denote the
tangential de Rham complex.%
\footnote{
    Formally, the global sections of 
    the graded sheaf $\bigwedge ^k(F^*)$.
}
Its cohomology groups are the \emph{tangential cohomology}
groups and are 
denoted $H_\tau ^k(X)$.  These vanish for $k > d$ because there are
no forms in higher dimensions.  There is a natural 
map  \cite[p.~58]{MS}
\begin{equation}
s\co   \CH^k(X ; \RR ) \longrightarrow H_\tau ^k(X)
\label{eq:comparison}
\end{equation}
from \v Cech cohomology to tangential cohomology, defined  
using sheaf theory.   If $X$ is a compact smooth foliated manifold $M$, then
this simply corresponds to the inclusion
$C^\infty(M) \subset C_\tau^\infty (M)$,
since any smooth function on $M$ is tangentially smooth.

The groups $H_\tau ^k(X)$ have a natural topology induced from the de
Rham cochains, and the topology is not necessarily Hausdorff. We
denote by $\bar H_\tau ^k(X)$ the Hausdorff quotient of $H_\tau^k(X)$. 

There is an associated \emph{tangential  homology} theory $H_*^\tau (X)$
defined by taking the homology of 
$ [\Omega _\tau ^*(X)]^*$,
the (continuous) dual of the associated  tangential 
de Rham complex. Then there are   natural isomorphisms 
\[
\Hom_{\textup{cont}} (H_\tau^k (X) , \RR ) \cong
\Hom_{\textup{cont}} (\bar H_\tau^k (X) , \RR ) \cong H_k^\tau (X).
\]

The comparison map $s$ of \eqref{eq:comparison} is for general
foliated spaces neither injective nor surjective.  But we have a
substantial simplification when our foliated spaces satisfy
Definition \ref{def:foliatedRd}. To explain it, recall that when
$N$ is a totally disconnected compact metrizable space,
$C^\infty(N)$ is the algebra of locally constant functions on $N$
(this notation figures prominently in analysis on $p$-adic groups),
which is the algebra of functions that factor through some
quotient map $N\twoheadrightarrow F$ with $F$ a discrete finite set.
This is a dense subalgebra of $C(N)$, of countable
dimension, and there is a natural
completion map with dense range $i\co C^\infty(N)\hookrightarrow C(N)$.

\begin{Thm}
\label{thm:TangentialDeRham}
Let $X$ be a compact foliated space given by an $\RR^d$-action
in the sense of \textup{Definition~\ref{def:foliatedRd}}, obtained
by inducing a $\ZZ^d$-action on a totally disconnected compact
space $N$.  Then
there is a natural commuting diagram
\[
\xymatrix{
\CH^k(X ; \RR ) \ar[r]^s \ar[d]^\cong 
& H_\tau ^k(X) \ar[r] \ar[d]^\cong &
\bar H_\tau ^k(X)  \\
H^k_{\textup{group}}(\ZZ^d, C^\infty(N)) \ar[r]^{i_*}
& H^k_{\textup{group}}(\ZZ^d, C(N)). &
}
\]
The composite
\[
\bar s\co  \CH^k(X ; \RR ) \xrightarrow{s}   H_\tau ^k(X) \longrightarrow  \bar H_\tau ^k(X)
\]
has dense image.
\end{Thm}
\begin{proof}
Recall that $H_\tau^*(X)$ is the sheaf cohomology of the sheaf
$\cR_\tau$ of germs of continuous real-valued functions which are locally
constant along the leaves, since the tangential de Rham complex is a fine
resolution of this sheaf, whereas $\CH^k(X; \RR)$ is the sheaf cohomology
of the sheaf $\cR$ of germs of locally constant real-valued functions.
These sheaves are not the same. But
\eqref{eq:comparison} is induced by the natural ``inclusion''
morphism $\iota\co\cR \to \cR_\tau$.  The bundle projection
$p\co X\to T^d$ gives rise to Leray spectral sequences of sheaves
\[
H^k(T^d, R^\ell p_*\cR) \Rightarrow H^{k+\ell}(X, \cR) \quad
\text{and} \quad
H^k(T^d, R^\ell p_*\cR_\tau) \Rightarrow H^{k+\ell}(X, \cR_\tau).
\]
Here $R^\ell  p_*\cF$ is the $\ell$-th derived push-forward of
a sheaf $\cF$, defined by ``sheafifying'' the presheaf
$U\mapsto H^\ell(p^{-1}(U), \cF)$.
The morphism $\iota$ induces a morphism of spectral sequences from the
first of the spectral sequences to the second.  The first spectral
sequence is just the familiar Serre spectral sequence
$H^k(T^d, \CH^\ell(N)) \Rightarrow \CH^{k+\ell}(X; \RR)$, though
note that we need cohomology with local coefficients here since
$\pi_1(T^d)=\ZZ^d$ acts nontrivially on $N$ and on
its only non-zero \v Cech cohomology group, $\CH^0(N;\RR)$.

Let's examine the sheaf $R^\ell p_*\cR_\tau$ which appears in the
second spectral sequence (the one converging to tangential
cohomology). For $U$ a small connected open set in $T^d$,
$p^{-1}(U)$ splits as $N\times U$, and the sheaf $\cR_\tau$ on this
open set of $X$ is just the sheaf of  germs of continuous functions
which are locally constant along the leaves, i.e., which depend only on
the $N$ factor in this product decomposition.  Thus
$H^\ell(p^{-1}(U), \cR_\tau) \cong H^\ell(N, \underline\RR)$,
where $\underline\RR$ is the sheaf of germs of real-valued continuous
functions, is actually
independent of $U$ (once it's small enough) and is just
$C(N)$, continuous real-valued functions on $N$,
for $\ell = 0$, and $0$ for $\ell > 0$. However,
\[
\CH^0(N;\RR)\cong \CH^0(N;\ZZ)\otimes_\ZZ \RR \cong
C(N, \ZZ)\otimes_\ZZ \RR \cong C^\infty(N),
\]
the locally constant functions on $N$, which is a dense subspace of
$C(N)$. So now it's evident that
\[
\iota_*\co H^k(T^d, R^\ell p_*\cR) \to H^k(T^d, R^\ell p_*\cR_\tau)
\]
is just the completion map
\[
i_*\co H^k_{\textup{group}}(\pi_1(T^d), C^\infty(N)) \to
H^k_{\textup{group}}(\pi_1(T^d), C(N))
\]
for $\ell=0$ and
that both sides vanish identically when $\ell>0$.
So both spectral sequences collapse and $\iota_*=i_*$.

As for the last statement about density of the image, it is obvious
that density of the image of $i\co C^\infty(N) \to C(N)$ gives
density of the image of $i_*$ in the topology of
$H^k_{\textup{group}}(\pi_1(T^d), C(N))$ coming from convergence
of (group) cocycles.  However, this topology also agrees with the topology
on $H^k_\tau(X)$ coming from the tangential de Rham complex, i.e., given
by convergence of differential forms, as one can see by restricting
to a small open set of the form $N\times U$, where the de Rham complex
locally looks like $C(N)\otimes \Omega^*(U)$.
\end{proof}

Because of Theorem \ref{thm:TangentialDeRham}, we see that for tiling spaces,
tangential and \v Cech cohomologies are close to being identical.
In our application we will need the composition
\[
\bar s \co  \CH ^d(X ; \RR)  \xrightarrow{s}
H_\tau ^d(X)    \longrightarrow  \bar H_\tau ^d(X) ,
\]
which under these assumptions will have dense image. 

In this regard
still an additional simplification occurs in the top degree, as
explained in the following section.


\section{Coinvariants}

In the more recent tiling literature there is quite a bit of attention
given to the  \v Cech cohomology group 
$\CH ^d(\Omega _T; \FF )$, since that is the group that houses the cohomology
information about the tiling.      
Whenever there is a fibration 
\[
N  \to X \to T^d
\]
as is the case with $X = \Omega _T$ usually, 
then
$\ZZ ^d \cong \pi _1(T^d) $ acts on    $N$ and hence on
$\CH^0(N ; \ZZ )$. Let 
$\CH^0(N ; \ZZ )_{\ZZ ^d}  $   denote the coinvariants of the action,
i.e., the quotient of $\CH^0(N ; \ZZ )$ 
by the subgroup generated by all of the elements
$g\cdot x - x$,
for $g\in {\ZZ ^d}$ and $x\in \CH^0(N ; \ZZ )$.
The following result seems to be well-known to the cognoscenti.

\begin{Thm}\label{coinvariant}
Fix some integer $d >0$ and let $T^d $ denote the $d$-torus. Suppose given a
fibration 
\[
N \longrightarrow X \longrightarrow  T^d
\]
with $N$ and $X$ compact and $N$ totally disconnected.
Then there is an isomorphism
\[
\CH ^d(X ; \ZZ ) \simeq      \CH^0(N ; \ZZ )_{\ZZ ^d}  ,
\]      
the coinvariants of the action of $\ZZ ^d \cong \pi _1(T^d)$ on
$\CH^0(N ; \ZZ )$.  Similarly with $\RR$ or $\QQ$ coefficients.
\end{Thm}
\begin{proof}
We have a Leray-Serre spectral sequence
$\CH^k(T^d, \CH^\ell(N ; \ZZ ))\Rightarrow \CH^{k+\ell}(X; \ZZ)$.
Here $T^d$ is the classifying space for $\ZZ^d$ and the outer
cohomology is just the same as group cohomology for the
fundamental group $\ZZ^d$ of $T^d$.
The spectral sequence collapses since $\CH^\ell(N ; \ZZ )=0$ for $\ell>0$.
So 
\[
\CH ^d(X ; \ZZ ) \cong H^d_{\textup{group}}(\ZZ^d, \CH^0(N ; \ZZ )).
\]
By Poincar\'e Duality, 
\[
H^d_{\textup{group}}(\ZZ^d, \CH^0(N ; \ZZ ))
\cong H_0^{\textup{group}}(\ZZ^d, \CH^0(N ; \ZZ )) = 
\CH^0(N ; \ZZ )_{\ZZ ^d}.%
\footnote{
    This isomorphism can be understood
    as follows.  The group cohomology of the group $\ZZ^d$ with
    coefficients in a module $M$ can be computed from a standard
    cochain complex $\left(\Hom(\bigwedge^k \ZZ^d, M), \delta\right)$.
    Thus $H^d(\ZZ^d, M)$ is the quotient of
    \[
    \Hom\left(\bigwedge^d \ZZ^d, M\right)\cong \Hom(\ZZ, M)
    = M
    \]
    by the image of $\delta$ from
    $\Hom\left(\bigwedge^{d-1} \ZZ^d, M\right)\cong \Hom(\ZZ^d, M)$.
    Examining the map shows that this is exactly the same as taking the
    coinvariants of the action on $M$.
    }
\]
The cases of $\RR$ or $\QQ$ coefficients
follow immediately by the Universal Coefficient Theorem.
\end{proof}  

\begin{Cor}\label{cor:topofH}
Suppose that $X$ is a compact foliated space given by
an $\RR ^d$-action.  Then there is a natural map
\[
\CH^0(N ; \RR  )_{\ZZ ^d} \longrightarrow H_\tau^d(X)
\]
whose image is dense in the Hausdorff quotient $\bar H_\tau^d(X)$.
\end{Cor}

\begin{proof} This is immediate from Theorems
\ref{thm:TangentialDeRham}  and \ref{coinvariant}.
\end{proof}

This corollary is important to us because of the relationship of
$\bar H_\tau ^d (X)$ to invariant transverse measures, as we will explain.


\section{Groupoids and Measures}

A compact foliated space $X$ has an associated {\emph{ holonomy groupoid}} $G(X)$, as in
\cite[p.\ 76, 77]{MS}.  
The \emph{holonomy group} of a point $x \in X$ is defined by 
\[
G_x^x = \{ h : hx = x\},
\]
which is, of course, the isotropy group of the action at the point
$x$. In general this group is countable. If the action of $\RR^d $ on
$X$ is free (\textbf{which we do not assume}), then each holonomy group will
be trivial.   Each point $x \in X$ lies on a unique leaf, which is the
orbit of the point $x$ under the action.  The leaf may be $\RR ^d$,
or some quotient torus, or something in between ($\RR^{d-k}\times T^k$).

If the foliated space is given by an $\RR^d $ action, as we
assume throughout (see Definition \ref{def:foliatedRd}), and if
the holonomy cover of each leaf is simply connected, then       
\[
G(X) \cong   \RR ^d \times X
\]
with unit space $X$, and associated map 
\[
\iota\co X \longrightarrow \RR ^d \times X \qquad \qquad \iota(x) =
(0,x),
\]
range and source maps given by 
\[
s(g,x) = x \qquad r(g,x) = gx ,
\]
and the inverse map given by 
\[
(g, x) ^{-1} = (-g, gx).
\]
Two elements $(g,y)$ and $(h,x)$ are multipliable if $y = hx$, and
then 
\[
(g, y) (h, x) = (g+ h, x) .
\]

Both $X$ and $G(X)$ are 
standard Borel spaces.  A \emph{transversal} $S$ is a Borel subset of $X$
that intersects each leaf in a countable (i.e., finite or countably
infinite) set. 
It is \emph{complete} if it intersects each leaf at least once.
Foliated spaces always have complete transversals.  Note that
transversals need not be    connected; for tilings they are
typically Cantor sets.  A transversal $N$ is \emph{open-regular} if
there is an open set $L \in \RR ^d$ and an isomorphism 
of foliated spaces of $L \times N$ onto an open subset of $X$ which
is the identity on $N$.  A transversal is \emph{regular} if it is
contained in an open regular transversal.  

An \emph{invariant transverse measure} on $X$ \cite[p.\ 82]{MS}
is a measure $\nu $ on the $\sigma $-ring  of Borel transversals $\SSS$ 
such that $\nu |_S $ is $\sigma $-finite for each $S \in \SSS $ and
$\nu |_S $ is invariant on  $G_S^S = \{(g,x) \in G(X) : x, gx \in S \}   $.
The key example for us will be the invariant transverse measure 
produced by a Ruelle-Sullivan current (\cite{RS}, though Ruelle-Sullivan worked exclusively with foliated manifolds).  It is important to note that 
not every foliated space has an invariant transverse measure, but the
hull of a tiling space does.  An 
invariant transverse measure $\nu $ is \emph{Radon} if
$\nu (N)$ is finite for each compact regular transversal $N$.
Basically all of the transverse measures we will deal with in this
paper are Radon.

A \emph{tangential measure} is a Borel assignment
$\ell \mapsto \lambda ^\ell $ of a (nonnegative, $\sigma$-finite)
measure $\lambda ^\ell $
to each leaf $\ell $. For example, if $X = L \times N$ then we may
simply take Lebesgue measure on each leaf $L \times \{n\} \cong \RR ^d$. 

Given an oriented tangential measure $\lambda $ and a
signed Radon invariant transverse measure $\nu $ on a compact foliated 
space $X$, we may define a new 
signed Radon measure $\mu $ on bounded Borel sets in $X$ by 
\[
\mu = \int \lambda\, d\nu .
\]
In particular, if $\sigma \in \Omega _\tau ^d(X)$ is a tangentially
smooth $d$-form and $o$ is an orientation, then $\sigma _1 = o\sigma $ is 
a tangentially smooth volume form. Restricting to a leaf defines a
signed measure with a $C^\infty $ density and hence a (signed)
tangential measure 
$\lambda _\sigma$.
Then $\int \lambda _\sigma \,d\nu $ is defined and its total mass gives a
real number. 

The integral can therefore be viewed as a linear functional 
\[
C_\nu \co \Omega _\tau ^d(X) \longrightarrow \RR
\]
where
\[
C_\nu (\sigma ) = \int _X \lambda _\sigma\,  d\nu .
\]
The fact that $\nu $ is invariant implies that $C_\nu $ is a closed
$d$-cycle and hence defines a homology class $[C_\nu ] \in H_d^\tau (X)$. 
This was first defined by Ruelle-Sullivan \cite{RS} and is called the
\emph{Ruelle-Sullivan current} associated with the invariant transverse
measure $\nu $.

Let $MT(X)$ denote the vector space of invariant Radon 
transverse measures on $X$.

\begin{Thm}[Riesz Representation Theorem {\cite[Theorem 4.27]{MS}}]
If $X$ is a compact oriented foliated space with leaf dimension $d$,
then the continuous linear 
functionals on $H_\tau ^d(X) $ can be identified as the invariant Radon
transverse measures.  In particular, the Ruelle-Sullivan map 
\[
C\co MT(X ) \longrightarrow \Hom_{\textup{cont}} (H_\tau^d (X) , \RR ) \cong H_d^\tau (X)
\]
which takes an invariant transverse measure $\nu $ to its Ruelle-Sullivan
current $C_\nu $ is an isomorphism of vector spaces.
\end{Thm} \qed

\medskip

\begin{Cor} 
\label{cor:tilingtransvmeas}
If $X$ is a compact foliated space given by an
$\RR ^d$-action with associated fibration $N \to X \to T^d$,
then every invariant transverse measure $\nu$ on $X$
is determined by the restriction of its Ruelle-Sullivan current
to the image of $\CH^0(N; \RR)_{\ZZ^d}$, or equivalently,
by the associated $\ZZ^d$-invariant linear functional on
$C^\infty(N)$.
\end{Cor}
\begin{proof} Apply Corollary \ref{cor:topofH}.
\end{proof}  

\medskip

\begin{Def} Let $K = \ZZ$ or $\RR $.
The \emph{cohomology trace}
\[
\tc\co \CH ^d(X; K ) \longrightarrow \RR
\]
associated to an invariant transverse measure $\nu $ 
is defined to be the composition
\[
\CH ^d(X ;K) \xrightarrow{\bar s} \bar H_\tau^d(X) 
\xrightarrow{\cap_\nu} \RR.
\]
\label{def:RS}
\end{Def}

\begin{Rem}
One case of special interest is when the foliation is 
\emph{uniquely ergodic}, or in other words there is a 
unique (up to scaling) invariant transverse probability measure $\nu$.
In the tiling space situation of Corollary \ref{cor:tilingtransvmeas},
that means that $N$ has a unique $\ZZ^d$-invariant probability
measure.  This situation was studied in \cite{MR3394105},
where in the context of Delone multisets with finite local
complexity, it was shown to be equivalent to
\emph{uniform cluster frequencies}. In the uniquely ergodic
case, since $MT(X)$ is one-dimensional, so is the Hausdorff
quotient $\bar H_\tau^d(X)$ of
$H_\tau^d(X)\cong C(N)_{\ZZ^d}$, and $C^\infty(N)$ must
surject onto $\bar H_\tau^d(X)$.
\end{Rem}


\section {$K$-theory review}

$K$-theory for $C^*$-algebras plays a central role in this story;
in this section we review the properties 
that we will need (see Blackadar's book \cite{B} for a
good reference). 

The group $K_0(A)$ is defined for  any $C^*$-algebra $A$. If $A$ is
unital then we take the union of all 
of the self-adjoint projections $p$ living in finite-dimensional
matrix rings over $A$, form the free abelian 
group on this set, and then quotient out by the subgroup generated
by setting $[p] = [q]$ if $p$ and $q$ are 
unitarily equivalent, $[p \oplus 0] = [p]$, and
$[p \oplus q] = [p] + [q]$. If $A$ is non-unital, then we form
the unitization $A^+$, which is a unital algebra containing
$A$ as an ideal of codimension $1$, so that for example
$C_0(X)^+ \cong C(X^+)$, when $X$ is a locally 
compact space with the one-point
compactification $X^+$.  We define
\[
K_0(A) = \ker [ K_0(A^+)\longrightarrow K_0(A^+/A) = K_0(\CC) =  \ZZ ]\,;
\]
see \cite[Chapter 5]{B} for more details.  It is easy to
see that a $C^*$-map $f\co A \to B$ induces a map 
$f_*\co K_0(A) \to K_0(B)$ and that $K_0$ is a covariant functor
from $C^*$-algebras to abelian groups.  We can 
define $K_1(A)$ via an analogous procedure using unitaries
instead of projections, or else just define 
\[
K_j(A) = 
K_0(A \otimes C_0(\RR ^j ))  
\]
for all $j$ and note
(fortunately) that Bott periodicity holds, so that 
\[
K_j(A) \cong K_{j+2}(A) .
\]
If $A$ is commutative and unital and hence of the
form $A \cong C(X)$  then  $K_j(A) \cong K^{-j}(X)$, which 
is the classical $K$-theory for compact
topological spaces. (For $K$-theory for compact and locally compact spaces, 
see Atiyah \cite{Atiyah}.) 

$K$-theory and cohomology are related via the classical Chern character map 
\[
\ch \co K^{-d}(X) \longrightarrow \CH^{**}(X ; \QQ )\,,
\]
where  $\CH^{**}(X ; \QQ)$ denotes the sum of the even or
odd rational \v Cech cohomology groups of $X$, matching the parity of $d$.
This was defined by Chern initially using differential forms but for us
the simplest way is via
characteristic classes as in Karoubi \cite[pp.\ 280--284]{Karoubi}.
The Chern character becomes an isomorphism after tensoring with $\QQ$:
\[
\ch \co K^{-d}(X)\otimes \QQ  \xrightarrow{\cong}
\CH^{**}(X ; \QQ),
\]
and similarly, of course, for real coefficients. 

Projecting 
to  $\CH^{d}(X ; \RR )$ gives a map that we denote 
\[
\ch_d  \co K^{-d}(X) \longrightarrow \CH^d (X ; \RR ). 
\]

\begin{Pro} Suppose that $X$ is a compact foliated space
given by an $\RR ^d $-action, and assume the holonomy cover of each
leaf is simply connected {\lp}automatic if the action is free{\rp}. Then    
\[
C^*(G(X)) \cong C(X) \rtimes \RR ^d .
\]
\end{Pro}

\begin{proof} This is immediate since
$G(X) \cong \RR ^d \times X$ as topological groupoids. 
\end{proof}

\begin{Def} (Bellissard) The {\emph{non-commutative Brillouin zone}} associated 
to a compact foliated space given by an $\RR ^d$-action is 
\[
\zone =   C^*(G(X)) \cong C(X) \rtimes \RR ^d .
\]
\end{Def} 
We use the letter $\zone $ to honor L\'eon Brillouin, who introduced
this concept for crystals.  When $X = \hull$   arises as the hull of a tiling, then we write $\zone _T$ for the associated non-commutative Brillouin zone.

It follows immediately that
\[
K_0(\zone ) = K_0( C^*(G(X))) \cong K_0(C(X) \rtimes \RR ^d ).
\]
Connes' Thom isomorphism theorem  
\cite{CT} implies that there is a canonical isomorphism
\[
\varphi : K_{ d}(C(X))       \xrightarrow{\cong}    
K_0(C(X) \rtimes \RR ^d )  \]
and of course 
\[
K_{d}(C(X))  \cong  K^{-d}(X)
\]
for any compact space $X$.

Putting these isomorphisms together yields 
\begin{Pro} Suppose that $X$ is a compact foliated space
given by an $\RR ^d $-action, and assume the holonomy cover of each
leaf is simply connected. Then there is a natural 
sequence of 
isomorphisms
\[
K^{-d}(X) \cong K_d(C(X))   \cong
K_0(C(X) \rtimes \RR ^d) \cong K_0(C^*(G(X)) \equiv  K_0(\zone ).
\]
We let $\chi : K^{-d}(X) \to K_0(\zone )$ denote the composite
isomorphism.
\end{Pro}


\section{Traces on Groupoid $C^*$-algebras}

Suppose given an invariant transverse measure $\nu $  and a tangential measure $\lambda $ on a foliated space $X$.
Then $\mu = \int \lambda \, d\nu $ 
is a measure on $X$ and turns $G(X)$ into a \emph{measured groupoid}
(see \cite[p. 142]{MS}) by defining a measure $\tilde\mu $ on $G(X)$ by 
\[
\tilde \mu (E) = \int _X  \lambda ^x (E \cap G(X))^x \, d\mu (x) .
\]
Form the Hilbert space with associated direct integral decomposition
\[
L^2(G(X), \tilde \mu ) =
\int^\oplus L^2(\ell , \lambda ^\ell ) \,d\tilde \mu.
\]
Then \cite[p.\ 142]{MS} we may form the $*$-algebra of integrable
functions $L^1(G(X), \tilde\mu ) $ 
with natural $*$-representation
\[
\pi \co L^1(G(X), \tilde\mu )
\longrightarrow \mathcal{B} (L^2(G(X), \tilde \mu )) .
\]
Define  the von Neumann algebra  $W^*(G(X), \tilde \mu )$ to be the weak
closure of  $\pi (L^1(G(X), \tilde\mu ))$   in 
$\mathcal{B} (L^2(G(   X), \tilde \mu ))$.   
This is a Type II von Neumann algebra and will be a factor
if $C^*(G(X))$ is simple, which is the case if and only 
if the action of $\RR^d $ on $X$ (or if the lattice $\Lambda$ on $N$,
in the situation of Definition \ref{def:foliatedRd})
is minimal.%
\footnote{
    This fact was originally called the
    Effros-Hahn Conjecture and was proven in full generality in
    \cite{GR}.
}
If the action of $\RR^d $ on $X$ is only topologically
transitive (i.e., there is a dense orbit, but not every orbit need be
dense), then  $C^*(G(X))$ will be primitive and $W^*(G(X), \tilde\mu)$
will still be a factor if the transverse measure has full support.

There is a natural map 
\[
C^*(G( X))  \longrightarrow W^*(G( X), \tilde \mu ).
\]      
For any leaf $\ell $  we have the local representation
\[
\pi _\ell \co C^*(G( X)) \longrightarrow \mathcal B (L^2(\ell , \lambda ^\ell )).
\]     
Write $m^x = \pi _\ell(m) $ for $x \in \ell $.  We wish to define
\[
\phi _\nu \co C^*(G(X)) \longrightarrow \RR.
\]

This construction is
described in detail in  \cite[pp.\ 149--154]{MS}.  Let $m \in
C^*(G( X))^+ $.   As an  element of the von Neumann algebra, think of
$m = \{ m^x \}$ for $m^x \in \mathcal B (L^2(\ell, \lambda ^\ell )) $.
Then $m^x$ is  a positive operator on $L^2(\ell, \lambda ^\ell )$ and
it is \emph{locally traceable} in the sense of \cite[p.\ 18]{MS}.
Define a measure $\lambda _m (\ell ) $ on $\ell $ by the formula
\[
\int _\ell f \, d\lambda _m(\ell ) =  \Tr\,(f^{1/2} m^x f^{1/2} )
\qquad \forall  f \in L^2(\ell, \lambda ^\ell )
\]
for every positive $f$ of bounded support.  (Note that this is the same
as the construction of D. Lenz, N. Peyerimhoff, and I. Veseli\'c
in \cite[Theorem 4.2]{LPV}.)

Finally, define the trace itself by 
\[
\phi _\nu (m) = \int _\ell  \lambda _m(\ell )\, d\nu(\ell ).  
\]

If the measure $ \mu =   \lambda _m  d\nu  $ is finite on $ X$ and
Radon on $G( X)$ then for any $g \in C_\tau ^\infty (G(\Omega_T))$
with compact support, we have 
\[
\phi _\nu (g^*g) <   +\infty.  
\]   
The trace is thus densely defined. It is lower semi-continuous
\cite[pp.\ 149--154]{MS}.
              
In our situation where $G(X) \cong X  \times \RR ^d $ as topological
groupoids, we recall first that if $\lambda $ is a tangential measure
on $X$ and $\nu $ is an invariant transverse measure then 
\[
\mu = \int _\ell \lambda _\ell \,d\nu   
\]
is a Radon measure on $X $ which is $\RR ^d$-invariant, and
$\tilde \mu $ is a Radon measure on $G(X)$,   so that
there is a natural trace 
\[
\tau _\mu (f) = \int _{x \in \Omega _T} f(x, 0) \int _\ell \lambda _\ell \,d\nu.
\]

To fit this into our framework, we would first represent
$C(X) \rtimes \RR ^d $ in the von Neumann algebra
$L^\infty(X, \mu) \rtimes \RR ^d$ associated to 
it via the direct integral procedure. Then the function
$f$ is sent to the family $\{f^x\}$, where $f^x $ is a 
bounded operator on $L^2$ of the associated leaf $\ell _x$.
This produces a local trace $\lambda _f^x $ on $\ell _x$, 
the leaf of $x$.  Then we would define 
\[
\phi _\nu\co C(X) \rtimes \RR ^d \longrightarrow \RR
\]
by
\[
\phi _\nu (f) = \int \lambda _f\, d\nu 
\] 
and then 
\[
\tau _\mu (f) =   \phi _\nu (f). 
\]
by \cite[Prop.\ 6.25 and the discussion at the bottom of page 149]{MS}.


\section{The Partial Chern character and $K$-theory trace}

There is a \emph{partial Chern character}
\[
c\co K_0(C^*(G(X)) ) \longrightarrow \bar H_\tau ^d (X)
\]
defined in \cite[p. 161]{MS}  as follows.
Suppose that $[u] \in K_0(C^*(G(X)) )$ is represented by $[e] - [f]$,
where $e, f \in M_n(C^*(G(X)^+))$ with common image in 
$M_n(\CC )$.
Let $\nu $ denote a positive Radon invariant transverse measure 
on $X$  and form the corresponding trace $\phi _\nu $ on $C^*(G(X))$.
Extend  $\phi _\nu $ to $\phi _\nu ^n =    \phi _\nu \otimes \Tr $
on $M_n (C^*(G(X)))$.    Then $c[u] \in \bar H_\tau ^d (X) $ is
the cohomology class 
of the tangentially smooth $d$-form $\omega _u$ which (after
identifying $d$-currents with Radon invariant transverse measures), is 
given by%
\footnote{
    Here is the detail:    We start by noting that    
    \[
    K_0(C^*(G(X)) ) \cong K_0(C^*(G_N^N)) 
    \]
    and so without loss of generality we may assume that $e$ and $f$ are
    in  $M_n(C^*(G_N^N)^+)$. 
    Skipping some analysis (see \cite[p.\ 162]{MS}), we may assume
    that any element in $C^*(G_N^N)$ may be represented there by a
    kernel operator where the 
    kernel is continuous and has compact support. Further, the kernel,
    when extended to $G(X)$, is tangentially smooth. 
    Following through with a little more analysis, we see that we may
    express the action of the trace on any element $b \in C^*(G(X))$ as 
    \[
    \phi _\nu (b) = \int \omega _b\, d\nu 
    \]       
    We see from this analysis that the partial Chern character is given by 
    \[
    c[u] = [ \omega _{e - f}  ]
    \]
    and hence 
    \[
    \phi _\nu (e - f) = \int \omega _{e-f}\, d\nu = \int c[u]\, d\nu .
    \]
}
\begin{equation}\label{E:pchern}
\omega _u(\nu ) =  \phi _\nu ^n (e - f)
\end{equation}
where $\phi _\nu ^n $ is the trace $\phi _\nu  $ on  $C^*(G(X)) $
associated to the invariant transverse  measure.  

It is now easy to describe the $K$-theory trace.  
\begin{Def} 
Given an invariant
transverse measure $\nu$, the {\emph{$K$-theory trace}}, 
denoted 
\[
\tk \co K_0(G(X))  \longrightarrow \RR 
\]
is   the composition
\[
K_0(G(X)) \xrightarrow{c}
\bar H_\tau ^d (X) \xrightarrow{\cap C_\nu } \RR.
\]
If $X$ is a compact foliated space given by an
$\RR ^d$-action with invariant transverse measure $\nu $,
then we may write 
\[
\tk\co K_0(\zone ) \equiv K_0(G(X))  \longrightarrow \RR. 
\]
\end{Def}


\section{Uniting the Traces}
\label{sec:uniting}
The goal of this paper, as explained in the introduction, is to
demonstrate that the $K$-theory and cohomology approaches to the traces 
are related, and to show that they are equivalent in 
low dimensions. We can now state our main result. 
\begin{Thm}
\label{thm:main}
Suppose that $X$ is a compact foliated space given by an
$\RR ^d$-action with invariant transverse measure $\nu $,
and the holonomy cover of each leaf is simply connected. Then:
\begin{itemize}
\item 
The diagram 
\begin{equation}
\label{eq:main1}
\xymatrix{
K_0(\zone)\ar@2{-}[r]& K_0(C^*(G(X))) \ar[r]^(.6){\chi ^{-1}}_(.6)\cong \ar[d]^c
& K^{-d}(X) \ar[r]^{\ch_d} & \CH^d(X; \RR )\ar[d]^s\\
&H_\tau^d(X)\ar[rr]^{\id}&&H_\tau^d(X)
}
\end{equation}
commutes. 
\medskip

\item Bloch Theorem: For every invariant transverse measure $\nu$, the diagram 
\begin{equation}
\label{eq:main2}
\xymatrix@C+6ex{
K_0(\zone)\ar[r]^(.45){(\ch_d)\circ(\chi^{-1})} \ar[d]^{\tk}
& \CH^d(X; \RR)\ar[d]^{\tc}\\
\RR \ar[r]^{\id} & \RR
}
\end{equation}
commutes. 
\end{itemize}
\end{Thm}      

\begin{proof}
We gratefully acknowledge    Kaminker-Putnam \cite[Prop. 2.4]{KamP} for putting us on the right track for this 
theorem. 

Starting at $K^{-d}(X)$ and moving left one obtains  
composition
\[
K_d(C(X))  \xrightarrow{\varphi }        K_0 (C(X) \rtimes \RR ^d)    \xrightarrow{\cong}    K_0(C^*(G(X)))     
\xrightarrow{c}    \bar H_\tau ^d(X) 
\]
and this
is the abstract analytic index map $\index _a$ of A.\ Connes as described in
\cite{C-CR}. 
In the other direction, the composition 
\[
K_d(C(X))  \cong   K^{-d}(X)  \xrightarrow{\ch_d }
\CH ^d(X ; \RR )  \xrightarrow{s}   \bar H_\tau ^d(X) 
\]
is the abstract topological index $\inde_t$.     Connes shows
\cite[Theorem 9]{C-CR} that 
\[
\inde_a = \inde_t   \in   \bar H_\tau ^d(X)   .
\]
This is an early version of the 
abstract foliation index theorem of Connes and Skandalis.  

The commutativity of the second diagram follows at once
from applying the definitions of the traces. We expand the diagram
\begin{equation}
\xymatrix{
K_0(\zone)\ar@2{-}[r]& K_0(C^*(G(X))) \ar[r]^(.6){\chi ^{-1}}_(.6)\cong \ar[d]^c
& K^{-d}(X) \ar[r]^{\ch_d} & \CH^d(X; \RR)\ar[d]^s\\
&H_\tau ^d(X)\ar[rr]^{\id}\ar[d]^{\cap C_\nu}&&H_\tau ^d(X)\ar[d]^{\cap C_\nu}\\
&\RR\ar[rr]^{\id}&&\RR}
\label{eq:Chdiagram}
\end{equation}
and then observe that 
\[
\tk = (\cap C_\nu ) \circ c
\]
and 
\[
\tc = (\cap C_\nu ) \circ s 
\]
This, then, is essentially a special case of the Index Theorem 
for foliated spaces \cite{MS}.
\end{proof}

\medskip

Now we apply this result to tilings.  
Given a tiling $T$, we denote its continuous hull by $\Omega _T$.
Suppose that for any $R > 0$ 
that there are, up to translation, only finitely many patches in $T$
(i.e., subsets of $T$) whose 
union has diameter less than $R$. Then by \cite[Lemma 2]{RW},
$\Omega _T$ is compact. This condition is called the
\emph{finite pattern condition} or \emph{finite local complexity}.    

\smallskip
We assume the following conditions:
\begin{enumerate}
\item $T$ satisfies the finite pattern condition
(i.e., finite local complexity);
\item $T$ has only finitely many tiles up to translations,
meeting full-face to full-face,
and each of these tiles can appear in only finitely many orientations.
\end{enumerate}
Then by Sadun-Williams \cite[Theorem 1]{SW}, under these conditions, $\Omega _T$
is homeomorphic to the total space of a fibre bundle of the form 
\[
N \longrightarrow N \times _{\ZZ ^d} \RR ^d \longrightarrow T^d,
\]
obtained by suspending a ${\ZZ ^d}$-action on a totally disconnected
space $N$. The proof of this result depends on \cite[Lemma 4]{SW}, 
which asserts that this holds for rational tiling spaces.
While the proof of this Lemma in \cite{SW} is somewhat condensed,
Ian Putnam has explained it to us as follows.

Start with
a rational tiling space and rescale so that a translate of the tiling
has the property that all vertices
of tiles are in  $\ZZ^d$ (i.e., are at points with integral coordinates).
We have a space of tilings such that the vectors joining any two
adjacent vertices are in $\ZZ^d$. 
In consequence, if one vertex of a tiling is on an integer point, 
then they all are. Let $N $ be the set of all tilings in the space 
whose vertices lie in $\ZZ^d$.  First, it is compact. Second, 
because of finite local complexity, it is totally disconnected. 
(Fix a radius $R$. Look at all possible patches of radius $R$. 
There are only finitely many. This partitions $N $
into a finite number of closed disjoint sets. Let $R$ get bigger.) 
It also has an action of $\ZZ^d$ by translation. The
map from $N  \times \RR ^d$ to the tiling space  $\Omega _T$ which sends 
$(T, x)$ to $T-x$ induces a homeomorphism 
\[
h\co N \times _{\ZZ^d} \RR^d     \xrightarrow{\cong} \Omega _T.
\]
The properties claimed should 
be clear on $N \times_{\ZZ^d} \RR^d$. The map given by
Sadun and Williams is just projection onto the second component.

The fact that $\Omega_T $ is the total space of a fibre bundle
as described implies that the hull $\Omega _T$  is a compact
foliated space given by an $\RR ^d$-action.
So we can specialize the theorem above as follows.    

\begin{Thm}
\label{thm:maintiling}
Suppose that $T$ is a tiling satisfying the following conditions:
\begin{enumerate}
\item $T$ satisfies the finite pattern condition (i.e.,
finite local complexity);
\item $T$ has only finitely many tile orientations
\end{enumerate}
Then 
\begin{itemize}
\item 
The diagram 
\begin{equation}
\xymatrix{
K_0(\zone _T)\ar@2{-}[r]& K_0(C^*(G(\hull))) \ar[r]^(.6){\chi ^{-1}}_(.6)\cong \ar[d]^c
& K^{-d}(\hull) \ar[r]^{\ch_d} & \CH^d(\hull ; \RR )\ar[d]^s\\
&H_\tau^d(\hull)\ar[rr]^{\id}&&H_\tau^d(\hull)
}
\end{equation}
commutes. 

\medskip

\item Bloch Theorem: For every invariant transverse measure $\nu$, the diagram 
\begin{equation}
\xymatrix@C+6ex{
K_0(\zone_T)\ar[r]^(.45){(\ch_d)\circ(\chi^{-1})} \ar[d]^{\tk}
& \CH^d(\hull; \RR)\ar[d]^{\tc}\\
\RR \ar[r]^{\id} & \RR
}
\end{equation}
commutes. 
\end{itemize}
\qed
\end{Thm} 

\begin{Rem}    
Note that aperiodicity of $T$ is not needed for the commutativity.
But it is useful since it implies that $\RR ^d$ acts freely on $\Omega _T$.
\end{Rem}

\begin{Rem} Up to this point we have used {\Cech} cohomology with real coefficients. However, to understand the precise values of the trace,
and thus gap labelling, we need finer information, based on
{\Cech} cohomology with integer coefficients. These are related, of course, since 
$\CH ^*(X ; \ZZ)\otimes \RR \cong \CH ^*(X; \RR )$. The problem arises because of the Chern character. It is a map 
\[
\ch\co K^*(X) \longrightarrow \CH ^{**}(X ; \QQ ) 
\]
which induces an isomorphism
\[
\ch\co K^*(X)\otimes \QQ  \longrightarrow \CH ^{**}(X ; \QQ) 
\]
and thus also over the real numbers. 
For  arbitrary compact spaces the Chern character does NOT take only integral values. So in general there is no reason to think, for instance, that
\[
\ch_d\co K^{d}(X) \longrightarrow \CH ^d(X ; \QQ)
\]
factors through $\CH ^d(X ; \ZZ)$.

The good news is that this is the case in dimensions $\le 3$. 
(See for example \cite[Prop.\ 6.2]{AP}, though this is well known in 
the topology literature.  It was also observed indirectly in \cite{MR1269302}.)
The explanation of this is as follows. Complex line bundles over
a compact space $X$ are classified by a single invariant, the first
Chern class in $\CH^2(X;\ZZ)$. If $\dim X\le 3$, then every complex
vector bundle over $X$ is a direct sum of line bundles, and
one can define the integral Chern character
$\ch\co K^0(X)\to \CH ^0(X; \ZZ ) \oplus \CH ^2(X; \ZZ )$
by sending a virtual bundle (i.e., $\ZZ$-linear combination
of line bundles) to the ``rank'' in $\CH ^0(X; \ZZ )$ and the 
first Chern class $c_1$ in $\CH ^2(X; \ZZ )$.  This makes sense
even if $c_1$ is torsion, and defines a ring isomorphism
\[
\ch\co K^0(X)\to \CH ^0(X; \ZZ) \oplus \CH ^2(X; \ZZ).
\]
The case of $K^{-1}$ and spaces of dimension $\le 3$ is 
similar.  If $\dim X\le 3$, then $K^{-1}$ can be identified
with the homotopy classes of maps $X\to U(2)$.  Subtracting
off the class of a map $X\to U(1)$, which can be viewed as the
classifying map of a class in $\CH^1(X; \ZZ)$, we can 
assume we have a map $X\to SU(2)=S^3$, which for $\dim X\le 3$
is classified by an element of $\CH^3(X; \ZZ)$.  So we
get an isomorphism 
 \[
 \ch\co K^{-1}(X)\to \CH^1(X; \ZZ) \oplus \CH^3(X; \ZZ).
 \]
But for $X$ of dimension $4$ and up,
the Chern character involves $\frac{1}{2}c_1^2$ (and higher-order
terms in higher dimension) and so is only defined in
rational cohomology. When $\dim X=4$, the differentials
in the Atiyah-Hirzebruch spectral sequence (for computing $K$-theory
from {\Cech} cohomology) vanish, so that there is a filtration
of $K^0(X)$ with quotients $\CH ^0(X; \ZZ )$, $\CH ^2(X; \ZZ )$,
and $\CH^4(X; \ZZ)$, which is just a bit weaker than what
happens in lower dimensions. However, the extension in
recovering $K^0(X)$ from the cohomology can be non-trivial.
For example,  $\widetilde K^0(\RR\PP^4)\cong \ZZ/4$ while
\[
\CH^2(\RR\PP^4; \ZZ)\cong \CH^4(\RR\PP^4; \ZZ)\cong\ZZ/2.
\]
In this dimension the only denominator
needed to define the Chern character is $2$, 
so the Chern character can be viewed as a
ring homomorphism 
\[
K^0(X)\xrightarrow{\,\bigl(\rk,c_1,\frac{1}{2}c_1^2-c_2\bigr)\,} \CH ^0(X; \ZZ) \oplus \CH ^2(X; \ZZ)
\oplus \CH^4(X; \ZZ[\tfrac12]),
\]
which is an isomorphism after inverting $2$.
Anyway, for $\dim X\le 3$, we have isomorphisms
\[
\ch\co K^0(X)  \xrightarrow{\cong} \CH ^0(X; \ZZ ) \oplus 
\CH ^2(X; \ZZ ), \qquad K^1(X)  \xrightarrow{\cong} 
\CH ^1(X; \ZZ )  \oplus \CH ^3(X; \ZZ ).
\]
This will be used below.
\end{Rem}

\begin{Cor}[Bloch theorem in low dimensions]
\label{cor:maintiling}
Under the assumptions above, if the dimension is $\le 3$, then for every invariant transverse measure the diagram 
\begin{equation}
\label{eq:maintiling}
\xymatrix{
K_0(\zone_T)  \ar@{>>}[rr]^(.45){(\ch _d) \circ (\chi ^{-1})}  \ar[d]^{\tk}
&&\CH^d(\hull ; \ZZ ) \ar[d]^{\tc} \\
\RR \ar[rr]^{\id} &&\RR
}
\end{equation}
commutes. In particular,  the $K$-theory trace   
$\tk\co K_0(\zone_T) \to \RR$  and the cohomology trace  $\tc\co\CH^d(\hull;\ZZ) \to \RR$  have the same image in $\RR$.
\end{Cor} 
\begin{proof}
In  \eqref{eq:Chdiagram}, we can replace real cohomology with
integral cohomology in the upper right of the diagram.
\end{proof}
\begin{Rem}
In higher dimensions, the top degree part of the Chern character
$\ch_d\co K^{-d}(\hull) \to \CH^d(\hull;\QQ)$ does not obviously
factor through $\CH^d(\hull;\ZZ)$, but it does after multiplying
by $\lceil\frac{d}{2}\rceil!$\,.  So the images of the
$K$-theory trace and the cohomology trace agree at least up to this
factor.
\end{Rem}

\begin{Rem}
We take this opportunity to explain the connection between these results and the ``Gap Labeling Theorem'' (GLT) of \cite{BBG,KamP,MR2018220}.  In the situation of
Theorem \ref{thm:maintiling}, the GLT asserts that the image of
the $K$-theory trace on $K_0(\zone_T)$ is equal to the image of this trace
on $K_0$ of the subalgebra $C(\hull)$ of $C(\hull)\rtimes \RR^d$.  This is equivalent
to the equality of the images of the $K$-theory and cohomology traces, for the
following reason.  As pointed out in \cite[\S2]{KamP}, we can use the
equivalence between the groupoids $G(\hull)$ and $N\rtimes\ZZ^d$ (where $N$
is the totally disconnected transversal) to convert the statement of the GLT to
equality of the image of the $K$-theory trace on $K_0(C(N)\rtimes\ZZ^d)$ with
the image of the trace on $K_0(C(N))=K^0(N)=\CH^0(N;\ZZ)$ (since $N$ is
totally disconnected).  Since the transverse measure is assumed
invariant, the measure on $N$ is $\ZZ^d$-invariant, and this factors through
the coinvariants $\CH^0(N;\ZZ)_{\ZZ^d}$, which is $\CH^d(\hull;\ZZ)$ by Theorem \ref{coinvariant}.  In other words, the GLT asserts the
equality of the image of the $K$-theory trace with the image of
the cohomology trace on $\CH^d(\hull;\ZZ)$.  
By Theorem \ref{thm:maintiling}, the image of the
$K$-theory trace is equal to the image of the cohomology trace on
the subgroup of $\CH^d(\hull;\QQ)$ given by the image of $\ch_d$ on
$K^d(\hull)$.  So the images of the two traces are equal if
the image of $\ch_d$ is contained in $\CH^d(\hull;\ZZ)$.  Inclusion of the image
of $\ch_d$ in $\CH^d(\hull;\ZZ)$ is used in all of the proofs of the
GLT in \cite{BBG,KamP,MR2018220}.
\end{Rem}


\section{Building tilings : A quick introduction }

Having stated our main result, Theorems \ref{thm:main} and \ref{thm:maintiling} and specialized to tilings, we wish to apply it to show the 
relationship between structural (diffraction) data and spectral data. For the sake of self-completeness, we now present a quick introduction to some popular approaches used to build tilings. Further details can be found in \citep{Senechal_Book_1996,Luck_Review_1994,Baake_Grimm_Book_2013,Baake_Grimm_Book_2018,Mozes_JAM_1989,Janot_Book_1995,KatzDuneau_1986}.

\bigskip
\noindent \emph{Cut \& Project -- Characteristic function -- Phason}
\medskip

\noindent A commonly accepted view of a quasicrystal in dimension $d$ is modeled as a section of a periodic structure (crystal lattice $\ZZ^n$) in an $n$-dimensional ambient space
$\mathbb{R}^n$, with $n > d$. We have the decomposition $\mathbb{R}^n = E^\parallel \oplus E^\perp$, where $E^\parallel$ is the $d$-dimensional physical space in which the structure is embedded, whereas $E^\perp$ is an $(n-d)$-dimensional internal space. This setting is usually implemented for the  Cut \& Project algorithm (hereafter C\&P), very useful and popular for the building of quasicrystals~\citep{Senechal_Book_1996,Luck_Review_1994,DuneauKatz_1985,KatzDuneau_1986}.

A simple example is given by the quasiperiodic tiling of a line $(d=1)$ with a Fourier module $\mathcal F$ with two generators $n=2$. The ambient space is $\RR^2$ with the square lattice $\mathbb{Z}^2$; the physical space is the line $E^\parallel$, which makes a tilt angle $\theta$ with the horizontal axis. If the slope $s \equiv 1/ (1 +\cot \theta )$ is irrational, the structure thus obtained is quasiperiodic; it is a one-dimensional ($d=1$) quasicrystal. If the slope is an irreducible rational $s = p/q$, the structure is periodic with $q$ atoms in a cell. A celebrated example of one-dimensional quasicrystal is the Fibonacci sequence obtained for the irrational slope $s = \tau^{-1} = 2/ (1+ \sqrt{5})$.

For $n=2$, we end up with a deterministic arrangement of two types of tiles, i.e., a two-letter alphabet $\{a,b \}$ which generally represents a piecewise modulation of a physical parameter (e.g.\ density, potential, dielectric constant, etc.). The offset of the line cut fixes the first letter of the iteration. It is thus immaterial for the infinite tiling but not for finite chains. Since all choices are equivalent, the offset appears as a gauge freedom known as a phason and obtained by sliding the cut along the internal space $E^\perp$. 

A characteristic function, 
$ \chi (n, \phi) \equiv \sgn \, \left[ \cos \left( 2 \pi n \, s  + \phi \right) - \, \cos \left( \pi s \right) \right] $ 
with $n \in \NN$, equivalent to the C\&P algorithm can be defined, 
which takes the two values $\pm 1$ respectively identified to the two letters $\{a,b \}$.
The parameter $\phi \in [0, 2 \pi ]$ is the aforementioned phason serving as an extra gauge degree of freedom. 
$\chi (n, \phi) $ has been successfully used in the determination of both the diffraction spectrum~\citep{Luck_Review_1994,Senechal_Book_1996,Luck_PRB_1989} and spectral properties of Schr\"{o}dinger and of wave operators~\citep{Kraus_2012a,Kraus_2012b,Levy_arXiv_2015}.

\bigskip
\noindent \emph{Substitutions}
\medskip

\noindent Aperiodic tilings can also be generated by inflation rules known as substitutions~\citep{Luck_JPA_1993,Luck_Review_1994,Senechal_Book_1996,Baake_Grimm_Book_2018}. For a two-letter alphabet $\{a,b \}$, the substitution rule is defined by its action $\sigma$ on a word $w= l_1 l_2 \dots l_k$ by the concatenation $\sigma(w) = \sigma (l_1 ) \sigma (l_2 ) \dots \sigma (l_k )$. An occurrence primitive matrix 
$ M = \begin{psmallmatrix}
\alpha & \smash{\beta}\\
\gamma & \delta
\end{psmallmatrix}$
defined by $\sigma (a) = a^\alpha b^\beta$ and $\sigma (b) = a^\gamma b^\delta$ (ignoring the order of letters) is associated to $\sigma$. It allows us to define a sequence of numbers $F_N$ from the recurrence $F_{N+1} = t F_N - p F_{N-1}$, where $t = \Tr M$, $p = \det M$ and $F_{0,1} = 0,1$. The largest eigenvalue $\lambda_1$ of $M$ is larger than $1$ (Frobenius-Perron theorem).
For the Fibonacci substitution $ M =  \begin{psmallmatrix}1 & 1 \\
1 & 0
\end{psmallmatrix}$,
$\lambda_1 = \tau = \frac{1 + \sqrt{5}}{2}$ and $F_N$ are the Fibonacci numbers. 
The left eigenvector 
\begin{equation}
\mathbf{v}_1 = (\rho_a , \rho_b ),
\label{eq:rho_b}
\end{equation}
normalised to $\rho_a + \rho_b =1$, with $\rho_a = \frac{\gamma}{\lambda_1 +\gamma - \alpha}$, $\rho_b = \frac{\beta}{\lambda_1 + \beta - \delta}$, portrays the frequencies or densities of the letters $a$ and $b$ in the infinite word. 
The right eigenvector $\mathbf{w}_1 = (d_a,d_b)^T$, normalized such that $\frac{d_a}{d_b}=\frac{\beta}{\lambda_1-a}=\frac{\lambda_1-\delta}{\gamma}$, expresses the lengths of the corresponding tiles.

The C\&P and substitution algorithms are not equivalent; e.g.\ no substitution is associated to the C$\&$P slope $s = 1/\pi$ since $\pi$ is a transcendental and not an algebraic irrational. Conversely, the substitution $M = 
\begin{psmallmatrix}1 & 2\\
1 & 0
\end{psmallmatrix}$ has no C\&P counterpart.

Substitutions in $1d$, which are equivalent to C\&P, are quasiperiodic and can be identified as follows. 
Let $x_n$ be atomic positions on the boundary between tiles, $\bar{d}=\lim_{n\to\infty}x_{n}/n=\rho_{a}d_{a}+\rho_{b}d_{b}$ the mean tile length, and $u_{n}=x_{n}-\bar{d}n$ fluctuations about the mean. Then 
\[
\Delta_{u}=\lim_{n\to\infty}\sup u_{n}-\lim_{n\to\infty}\inf u_{n},
\]
is the extension of the atomic surface \citep{Luck_JPA_1993}. Then a substitution is quasiperiodic if $\Delta_{u}=1$. 
Furthermore, a substitution $\sigma$ is \emph{common unimodular} if it is primitive, irreducible, Pisot, unimodular ($\det M = \pm1$), and has a common prefix (or suffix) \citep{AR}. A quasiperiodic substitution is necessarily common unimodular.


\section{Diffraction spectrum}

In this section we discuss the diffraction spectrum 
information obtained from our examples of ($d=1$) one-dimensional tilings using {\Cech}  cohomology and Ruelle-Sullivan currents.

For a two-letter alphabet $\{a,b\}$, an atomic density $\rho (x) = \sum_n \delta (x - x_n)$ is defined by placing identical atoms at boundaries $x_n$ between $a$ and $b$ tiles. The structure factor or two-point correlation for a tiling of length (number of tiles) $N$ is 
\begin{equation}
S(k) = \frac{1}{N} |G(k)|^2 = \frac{1}{N} \sum_{m,n} e^{ik (x_m - x_n)},
\label{SF}
\end{equation}
where $G(k) = \sum_n e^{-i k x_n}$ is the Fourier transform of $\rho (x)$ and $k$ is the $1d$ wave vector in units of an inverse mean lattice spacing. 

Bragg peaks are essential in the definition of quasiperiodicity \citep{BT,hof1995,Dworkin}. There is a Bragg peak at $k_0$ if $G(k_0) \propto N$ for large $N$, namely if a macroscopic fraction of atoms diffracts coherently and $S(k) \simeq \delta (k - k_0)$. For a $1d$ quasicrystal described by the C\&P algorithm, 
\[
G(k) = \sum_{p,q} C_{pq} \, \delta \left( \tfrac{1}{2 \pi}k - p - q s^{-1} \right),
\]
hence the corresponding Fourier transform consists of Bragg 
peaks located at 
\[
k_{pq} = 2 \pi \, (p + q s^{-1}) \,,
\]
where $(p,q)$ are integers (see Fig.~\ref{fig:Fibonacci_diff} for a Fibonacci quasicrystal). The diffraction spectrum is pure-point, and the corresponding Fourier module $\ZZ + s\ZZ$ is the projection onto the physical space $E^\parallel$ of the reciprocal ambient space lattice~\citep{Luck_JPA_1993}. 

The diffraction spectrum of tilings generated by substitutions is obtained from the solutions of $k \lambda_1 ^n \rightarrow 0 \pmod{1}$, $\lambda_1 >1$ being the Perron-Frobenius eigenvalue \citep{BT}.  Non-vanishing $k$ solutions  depend on $\lambda_1$. 

A \emph{Pisot number} $\lambda$ of degree $p$ is a root of an irreducible monic polynomial $P(x)$ with integer coefficients, with $\lambda > 1$ and such that all other roots of $P(x)$ are less than  1 
in absolute value. For example, the golden ratio $\phi \sim 1.618$ is a root of $P(x) = x^2 - x - 1$. The other root 
of $P(x)$ is $-\phi ^{-1} \sim -0.618 $ which has absolute value less than $1$, so $\phi $ is a Pisot number.  

\medskip
We identify the following cases, which all fulfill conditions of Theorem \ref{thm:maintiling}: 
\begin{description}
\item [\boldmath$\lambda_1$ Pisot and \boldmath$\det M = \pm 1$.] The pure-point structure factor consists of Bragg peaks supported by a Fourier module with a finite number $p$ of generators. The structure is a quasicrystal and it can be described in the ambient space formalism. 

\item [\boldmath$\lambda_1$ Pisot and \boldmath$\det M \neq \pm 1$.] The pure-point structure factor consists of Bragg peaks supported by a Fourier module not finitely generated. For $1d$ chains, it contains at least the infinite family of Bragg diffractions $\{ 2 \pi  \lambda_1 ^{-n} , n \geq 0\}$. An example is provided by the Thue-Morse tiling whose occurrence primitive matrix is 
$ M= 
\begin{psmallmatrix}1 & 1\\
1 & 1
\end{psmallmatrix}$. The diffraction spectrum has both pure-point and singular continuous components. The corresponding Bragg peaks are computed in
\cite[Theorem 3.9]{MR2523312}.%
\footnote{
    There appears to be a slight misprint there; the $4\pi$ factor should apply to both summands in the cohomology. 
}
The corresponding Fourier module is $\ZZ+\ZZ\bigl[\frac{1}{2}\bigr]$. The full diffraction spectrum is
represented in Fig.~\ref{fig:Thue-Morse_diff} and detailed in Table~\ref{tab:examples}. Another example is the period doubling substitution, whose occurrence matrix is $ M= 
\begin{psmallmatrix}1 & 1\\
2 & 0
\end{psmallmatrix}$ with a pure-point diffraction spectrum of Fourier module $\ZZ \bigl[\frac12\bigr]$.   

\item [Non-Pisot substitutions] correspond to the case where the second eigenvalue of the occurrence matrix $M$ is $|\lambda_2 | >1$. This property has several consequences, among them the occurrence of unbounded density fluctuations~\citep{Godreche_Luck_PRB_1992,Dumont_Review_1990}. It has been shown that both the fluctuation, denoted by $u_n$, of the atomic positions in $E^\parallel$ and the extension of the tiling in the internal space $E^\perp$, scale with the power law $u_n \simeq n^\beta$ where $\beta = \ln |\lambda_2 | / \ln \lambda_1$. Such unbounded density fluctuations destroy the coherence of any would-be Bragg diffraction so that the Fourier diffraction spectrum of non-Pisot tilings is generically continuous. The Rudin-Shapiro tiling provides such an example as displayed in Fig.~\ref{fig:Rudin-Shapiro_diff}.
\end{description}

An alternative description of the Bragg spectrum for one-dimensional tilings is based on the \v Cech cohomology group $ \CH^{1}\left(\Omega_{T};\ZZ\right)$ associated to the hull $\Omega_T = \overline{\{ T - \mathbf{x} \mid \mathbf{x} \in E_\parallel \}}$.

A known result \citep{Julien_2009} is that for C\&P (quasiperiodic) tilings, 
$\CH^{1}\left(\Omega_{T};\ZZ\right) \cong \ZZ^a$, where $a$
counts the number of letters of the tiling.
The same applies to quasiperiodic substitutions.
The diffraction spectrum is obtained using the Ruelle-Sullivan map $C_\nu$ 
which projects $ \CH^{1}\left(\Omega_{T};\ZZ\right)$ into $\RR$,
\begin{equation}
C_\nu \left( \CH^{1}\left(\Omega_{T};\ZZ\right) \right) 
= \ZZ + \rho_b \, \ZZ,
\label{cech}
\end{equation}
so that Bragg peaks are labeled using the two integer coordinates of $\CH^{1}\left(\Omega_{T};\ZZ\right) \cong \ZZ^2$.%
\footnote{
    We have defined the Ruelle-Sullivan class $[C_\nu ]$ as lying in tangential homology $H_d^ \tau (X)$, so that it naturally pairs with tangential cohomology $H_ \tau ^d (X)$. We have constructed a natural map $s: \CH^d(X; \ZZ) \to H_ \tau ^d (X)$ and this induces a pairing of the Ruelle-Sullivan class with $\CH^d(X; \ZZ)$.  This composition is exactly 
    the cohomology trace $\tau _*^H : \CH^d(X; \ZZ) \to \RR $ as defined in \ref{def:RS}.
}

When $\det M\ne \pm 1$, then the {\Cech} cohomology often is
\emph{not} free abelian. For example, an analysis \cite[p.\ 531]{AP} 
or \cite{MR2383524} shows that 
for the Thue-Morse substitution tiling 
$\CH^1(\Omega _T; \ZZ) \cong \ZZ \oplus \ZZ[\frac 12]$.

\medskip

By Cor.~\ref{cor:maintiling}, we can calculate the range of the cohomology trace as well. So here is a summary of the results:
\medskip
\begin{Pro} ~
\begin{enumerate}
\item For the Fibonacci tiling, 
\begin{align*}
\CH ^1(\Omega_T ; \ZZ ) & \cong \ZZ \oplus \ZZ \, , 
 & \tc (\CH ^1(\Omega_T ; \ZZ )) & = \big(\ZZ + \lambda ^{-1}\ZZ \big)\, ,
\end{align*}
using $\rho_b=1-\lambda^{-1}$ in \eqref{eq:rho_b}.
\medskip
\item For the Thue-Morse tiling, 
\begin{align*}
\CH ^1(\Omega_T ; \ZZ ) & \cong \ZZ \oplus \ZZ \big[\tfrac 12\big] \, ,
& \tc (\CH ^1(\Omega_T ; \ZZ )) & = \tfrac 13 \, \ZZ\big[\tfrac 12\big] \, .
\end{align*}
\end{enumerate}
\end{Pro}


\begin{figure}[tb]
\centering
\subfloat[\label{fig:periodic_diff}Periodic.]
{\includegraphics[width=0.3\textwidth]{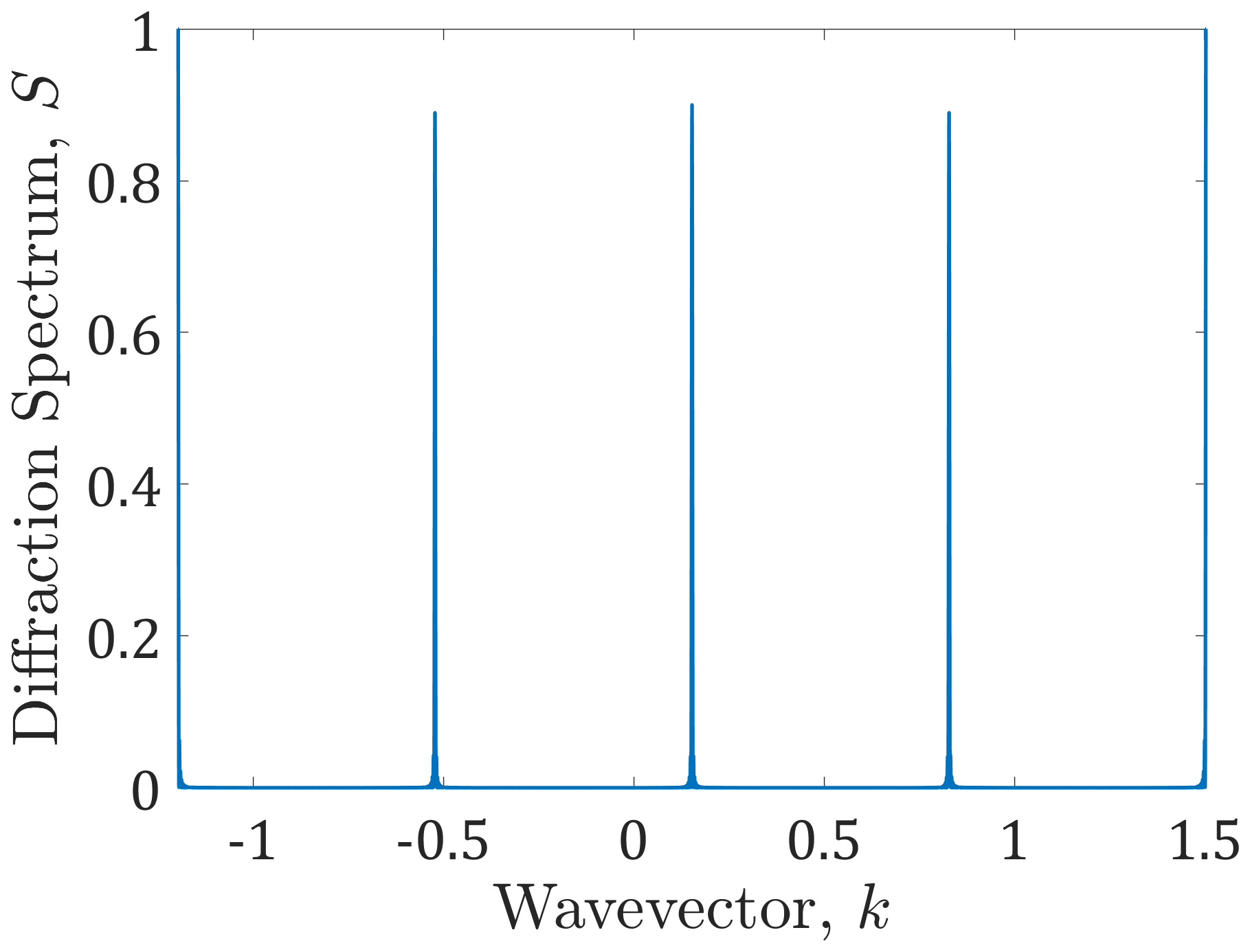} }
\hfill{}
\subfloat[\label{fig:Fibonacci_diff}Fibonacci.]
{\includegraphics[width=0.3\textwidth]{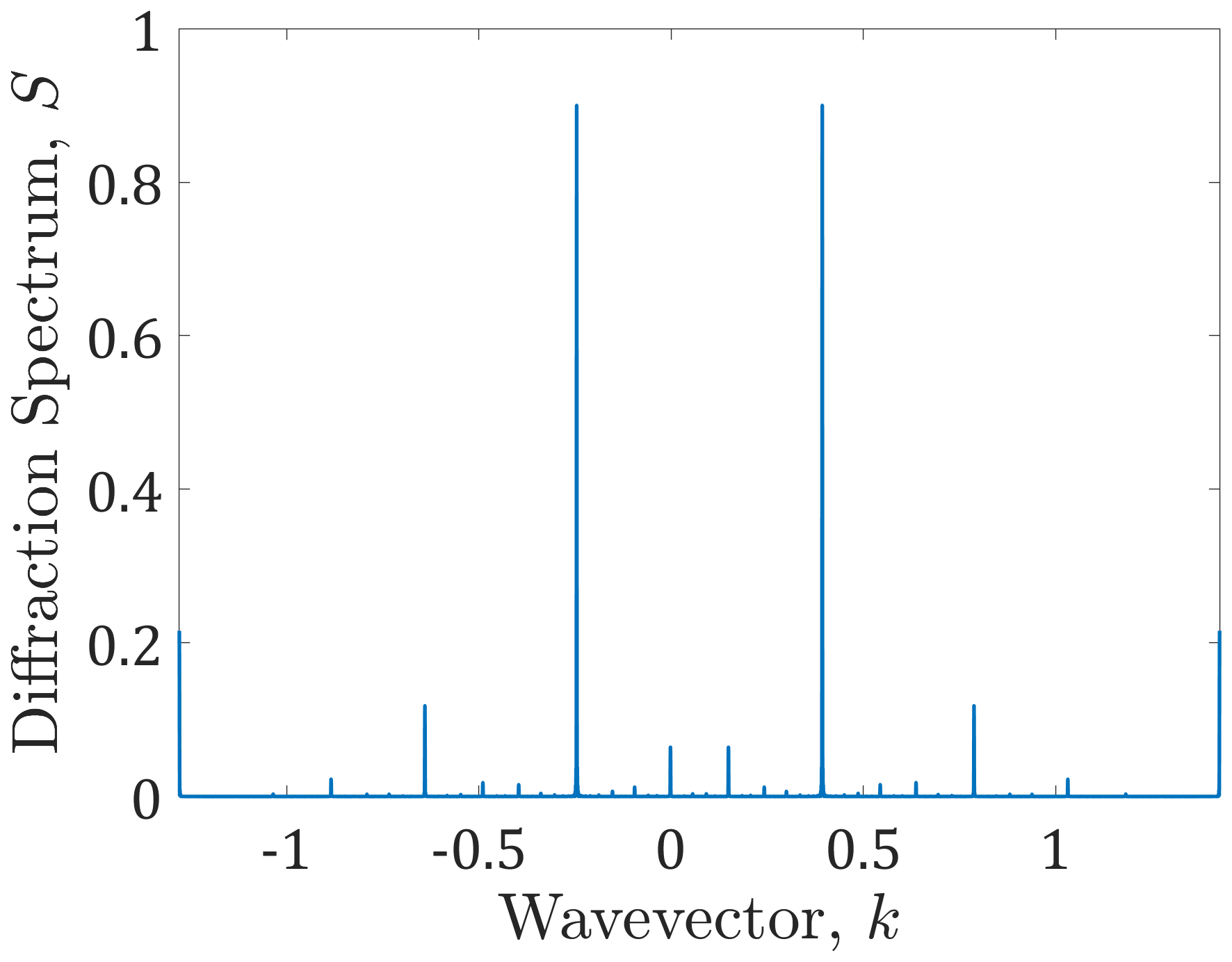} }
\hfill{}
\subfloat[\label{fig:Thue-Morse_diff}Thue-Morse.]
{\includegraphics[width=0.3\textwidth]{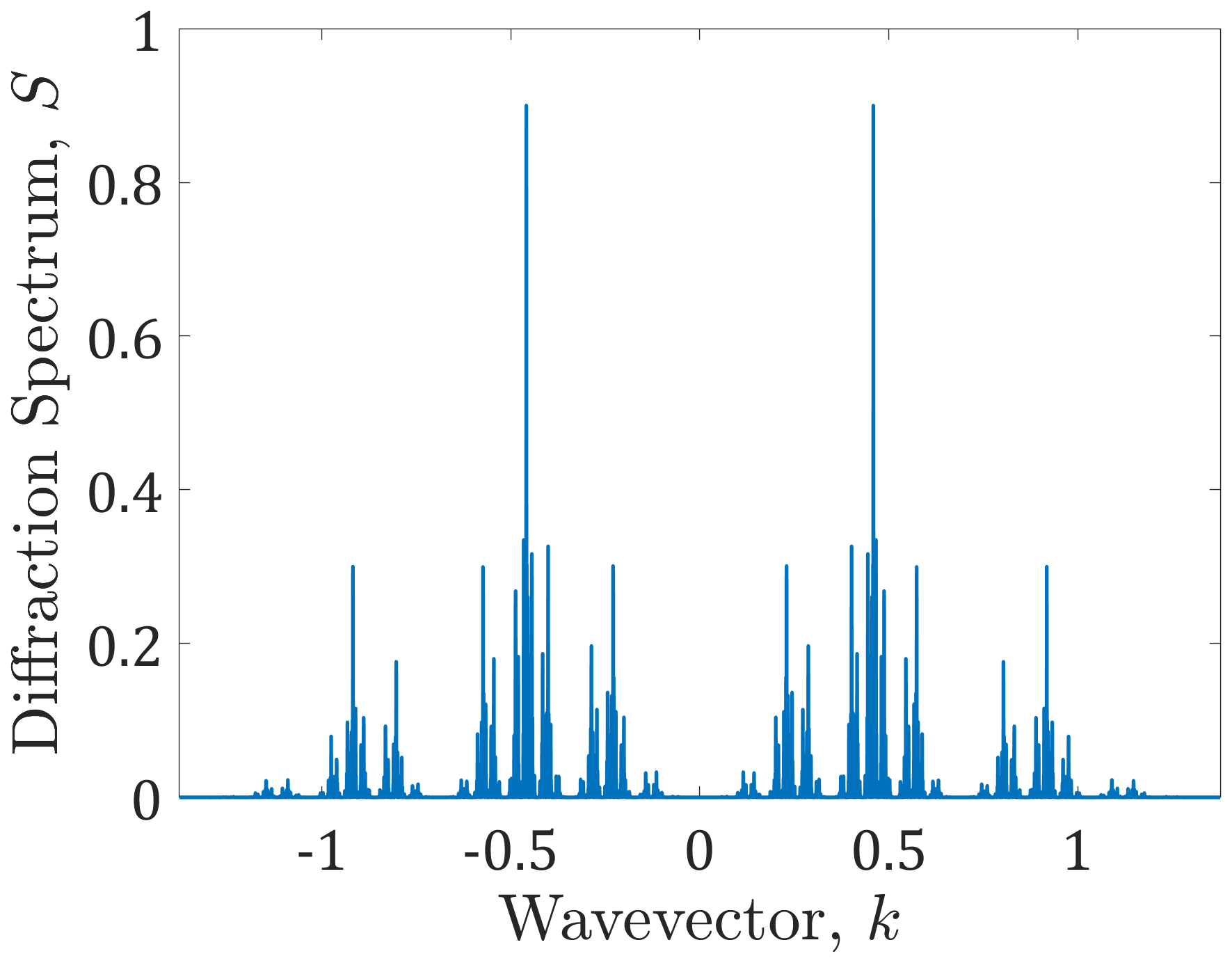} }
\\
\hfill{}
\hfill{}
\subfloat[\label{fig:Period-Doubling_diff}Period Doubl\rlap{ing.}]
{\includegraphics[width=0.3\textwidth]{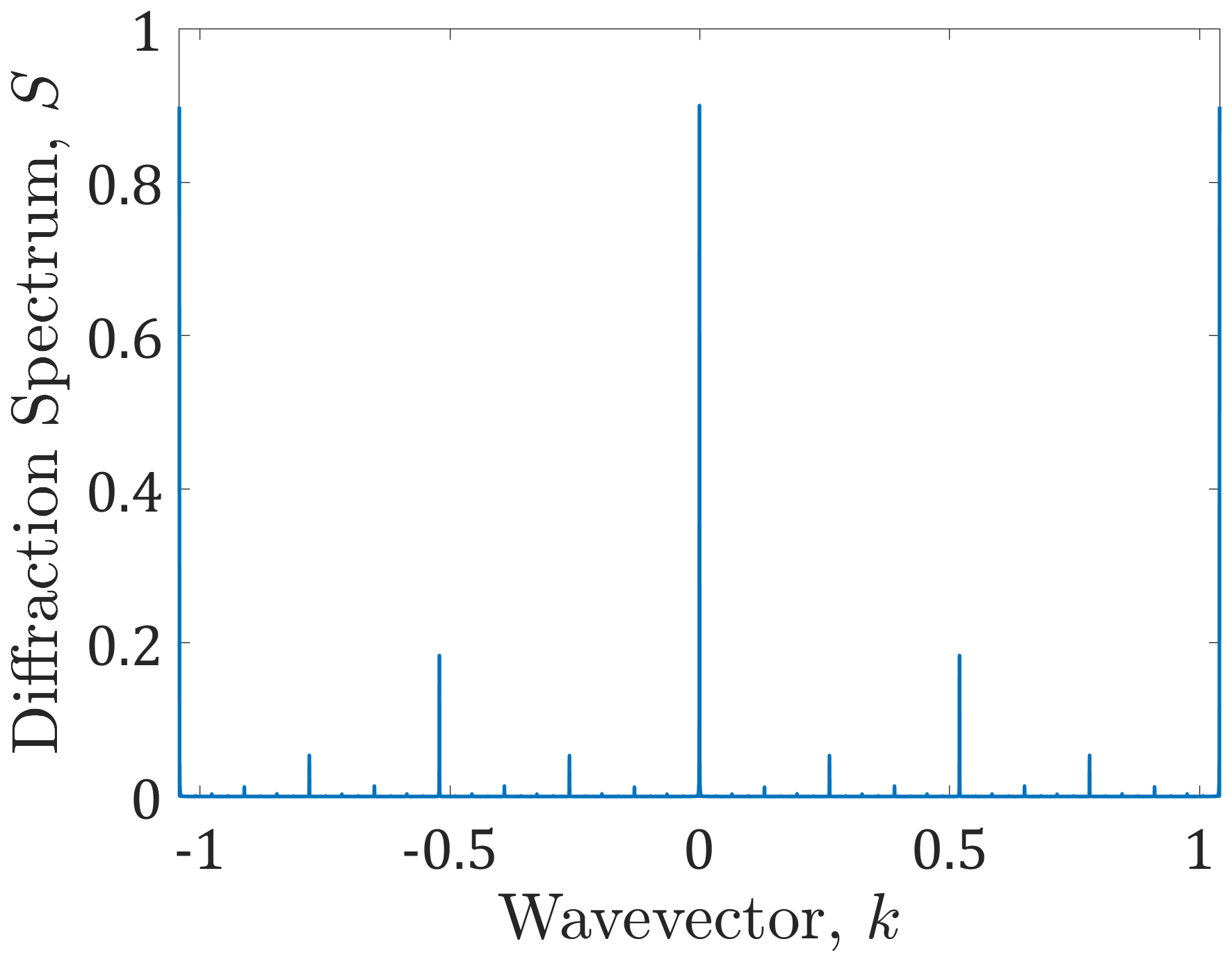} }
\hfill{}
\subfloat[\label{fig:Rudin-Shapiro_diff}Rudin-Shapiro.]
{\includegraphics[width=0.3\textwidth]{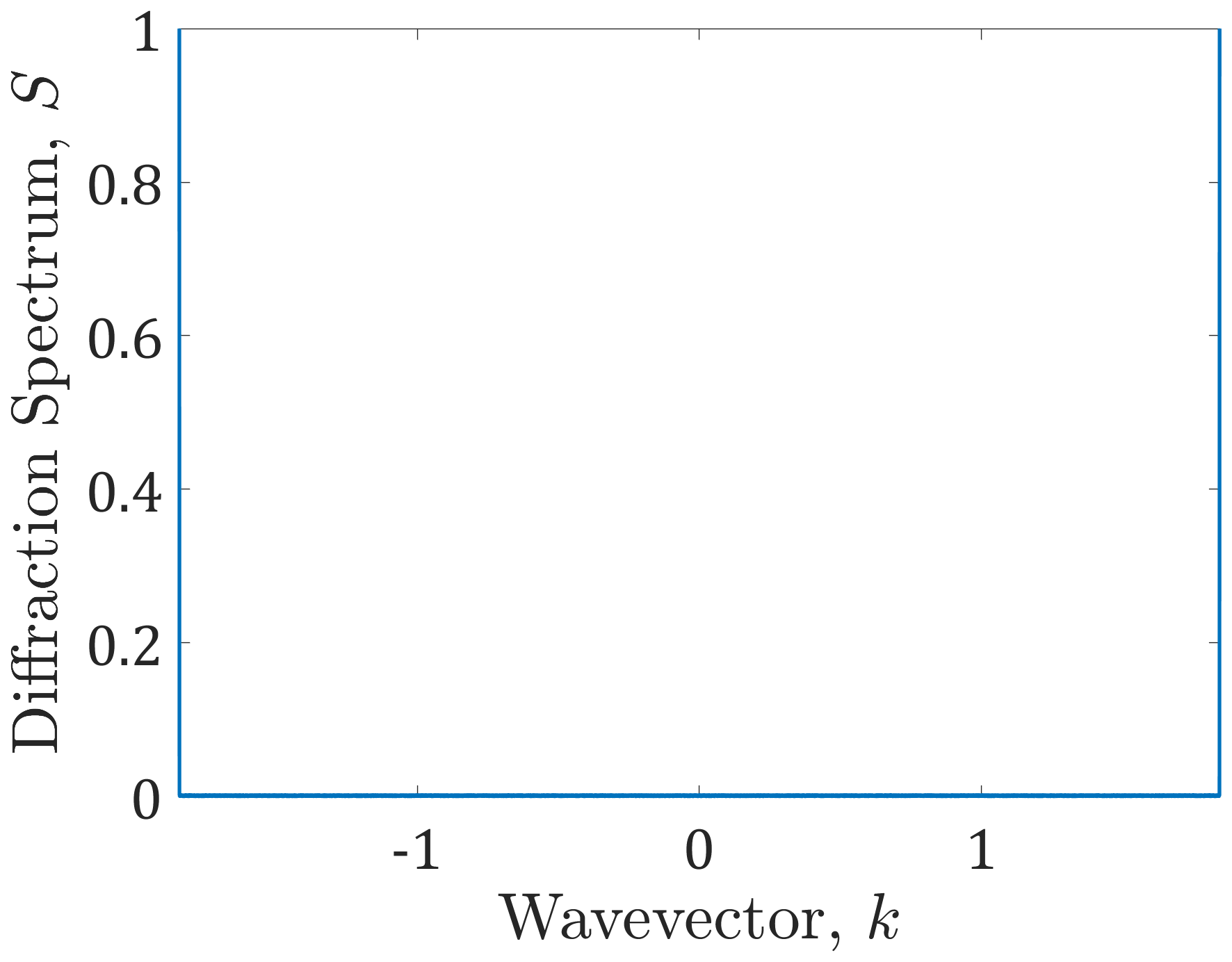} }
\hfill{}
\hfill{}
\caption{\label{fig:diff}
    Diffraction spectra of representatives of five families of tilings. (\textsc{a}) describes a periodic tiling with pure-point (PP) diffraction spectrum and a finite number of Bragg peaks, (\textsc{b}) a quasiperiodic (Fibonacci) tiling. The PP diffraction spectrum displays an infinite countable number of Bragg peaks, (\textsc{c}) a non-quasiperiodic tiling (Thue-Morse) with both PP Bragg peaks and a singular continuous (SC) component made of localised but not Bragg diffraction peaks (see Fig.~\ref{fig:diffraction-amp} for details), (\textsc{d}) a limit-quasiperiodic tiling (Period Doubling) with only PP Bragg peaks and (\textsc{e}) a non-quasiperiodic and non-Pisot tiling (Rudin-Shapiro) with an absolutely continuous diffraction spectrum. 
}
\end{figure}


\begin{figure}[ht]
\subfloat[Fibonacci.]{\includegraphics[width=0.48\textwidth]{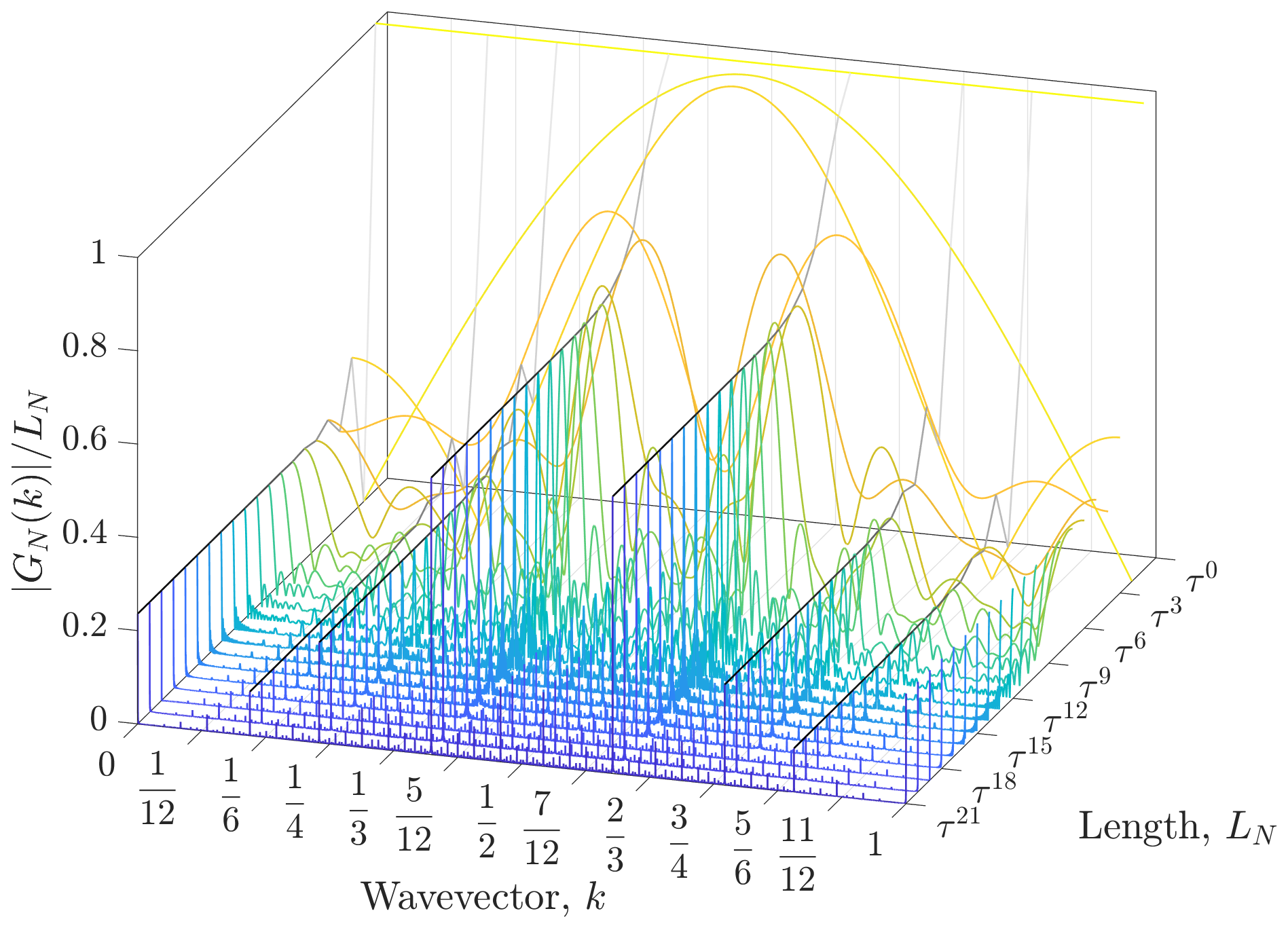} }
\hfill{}
\subfloat[Thue-Morse.]{\includegraphics[width=0.48\textwidth]{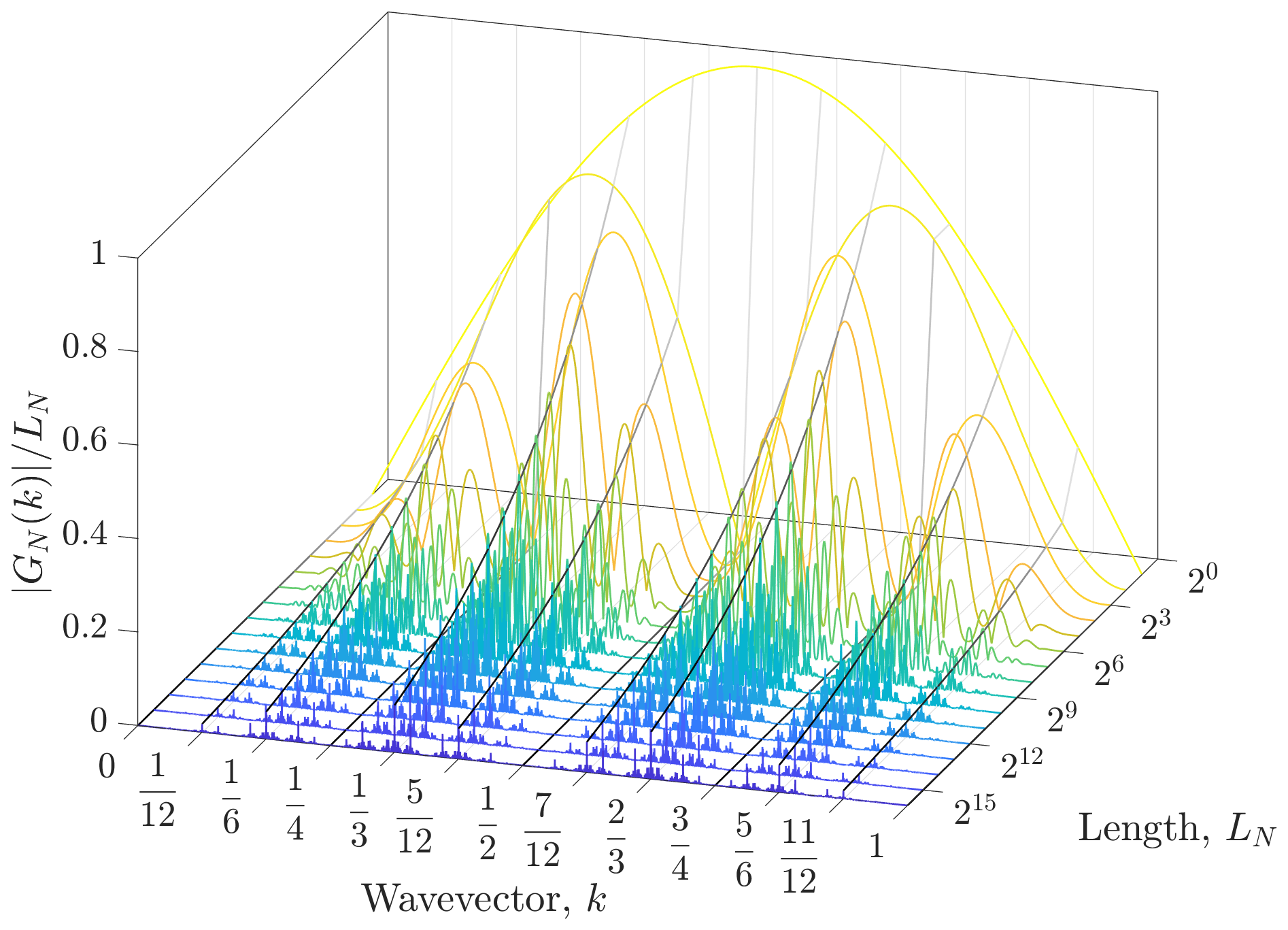} }
\caption{\label{fig:diffraction-amp} 
    Comparison between the PP discrete set of Bragg peaks and SC diffraction peaks. For Bragg peaks (e.g.\ Fibonacci), the structure factor in \eqref{SF} scales like the size $L_N$ of the $N$-letters tiling, whereas for SC diffraction spectrum (e.g.\ Thue-Morse), the structure factor scales like $L_N ^{\gamma_N}$ with $\lim_{N \rightarrow \infty} \gamma_N = \gamma < 1$.  The scaled diffraction amplitude $\left|G_{N}(k)\right|/L_{N}$ is represented. Note that for the PP case, the scaled diffraction amplitude saturates at large lengths, while for the SC case it decreases like the displayed power law ($\gamma = \log_2 3 -1$ for Thue-Morse).
    Color represents increasing order $N$. Black lines: the top of chosen peaks. 
}
\end{figure}


\section{Spectral properties and gap labeling theorem}

This section is devoted to spectral information extracted from the Laplacian operators conveniently defined on the previously discussed 
 representative classes of one-dimensional $d=1$ tilings, using $K$-theory 
and the partial Chern character.

The band structure for non-interacting excitations (e.g.\ electronic, electromagnetic, acoustic or mechanical waves) propagating in a tiling is modeled either by a ``tight binding'' model, where the tiles  $\{a,b \}$ represent atomic locations with particles hopping from tile to tile, or by a continuous wave equation.  Periodic
tilings model traditional crystalline structures. The quantum/wave  mechanical model of this motion is a certain self adjoint operator on the space of square-summable functions in
the set of tiles.  We are interested in the
spectrum of this operator (spectral data). 

The continuous versions of the Schr\"odinger and Helmholtz equations,
\begin{equation}
\frac{1}{2} \frac{d^2 \psi}{dx^2} - v(x) \, \psi = - k^2 \psi
\label{tb1}
\end{equation}
where $v(x)$ accounts for the tilings, have their advantages. The numerically more tractable (discrete) tight-binding version
\begin{equation}
\phi_{n+1} (e) +  \phi_{n-1} (e) + v_n \phi_n = 2 e \,\phi_n
\label{tb5}
\end{equation}
is extremely well documented in the condensed-matter physics literature \citep{Akkermans_Review_2014}. It is obtained from \eqref{tb1} by defining the dimensionless quantities $e = 1 - k^2 \ep^2$, $\phi_n (e) = e^{\ep^2 v_n / 2} \psi_n$ and $t_{n,n+1} = \exp \big( - \frac{\ep^2}{2} \, (v_n+ v_{n+1}) \big) $.
For a tiling of length (number of tiles) $N$, \eqref{tb5} can be rewritten in a matrix form $H_N \Phi = e \, \Phi$. The energy spectrum thus comprises $N$ eigenenergies denoted by $e_i$, $1 \leq i \leq N$. The counting function $\Nc(e)$ or integrated density of states is defined as the fraction of eigenenergies which are smaller than a given energy $e$, namely,
\begin{equation}
\Nc(e) = \frac{1}{N} \sum_{i=1}^N \theta \, (e - e_i ) 
\label{counting}
\end{equation}
where $\theta (x)$ is the Heaviside function. For large enough $N$, the counting function is independent of the choice of boundary conditions and it is usually a well defined and continuous function of energy. The counting function $\Nc(e)$ is represented in Fig.~\ref{fig:count} for different types of tilings.  


\begin{figure}[tb]
\centering
\subfloat[\label{fig:periodic_count}Periodic.]
{\includegraphics[width=0.3\textwidth]{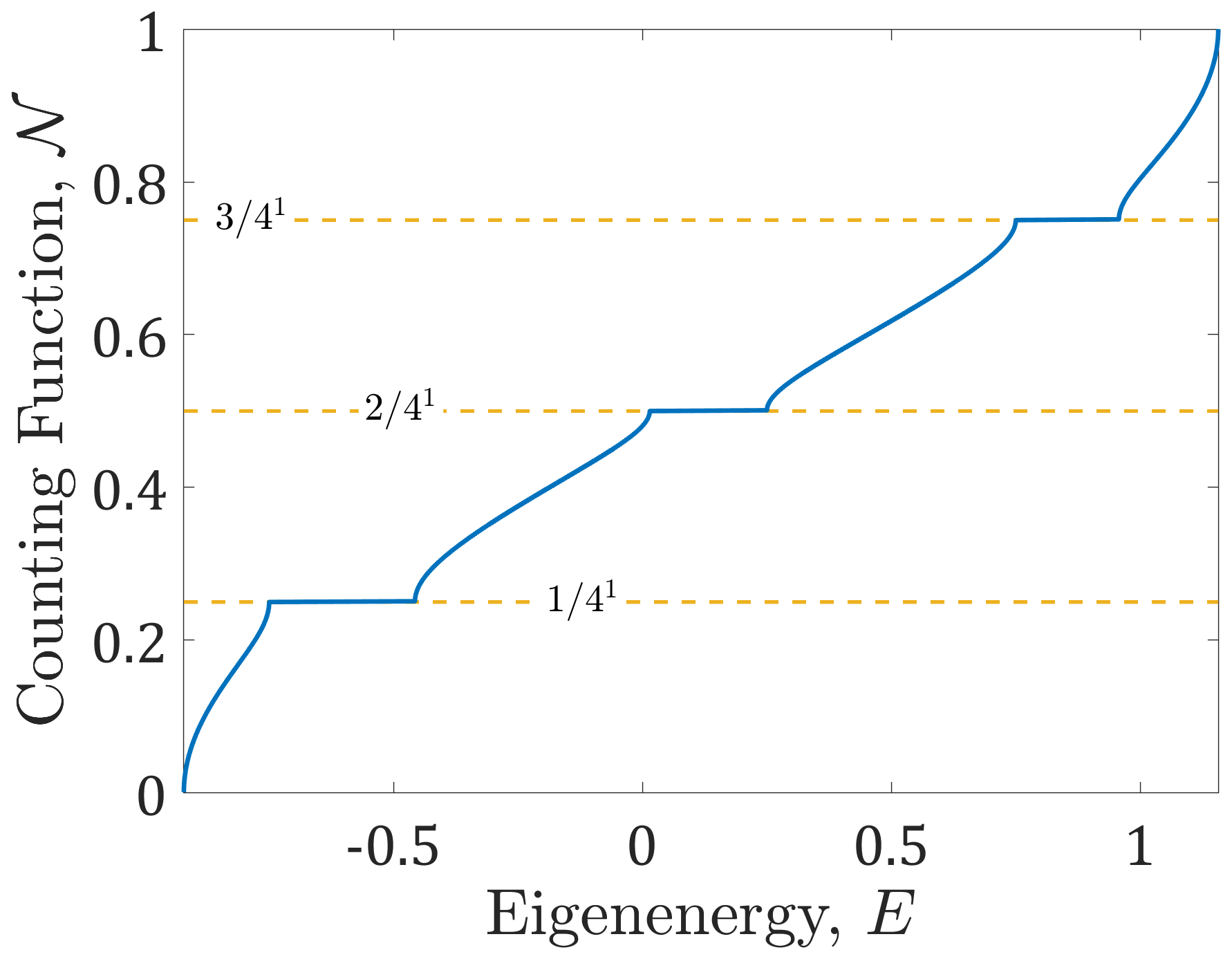} }
\hfill{}
\subfloat[\label{fig:Fibonacci_count}Fibonacci.]
{\includegraphics[width=0.3\textwidth]{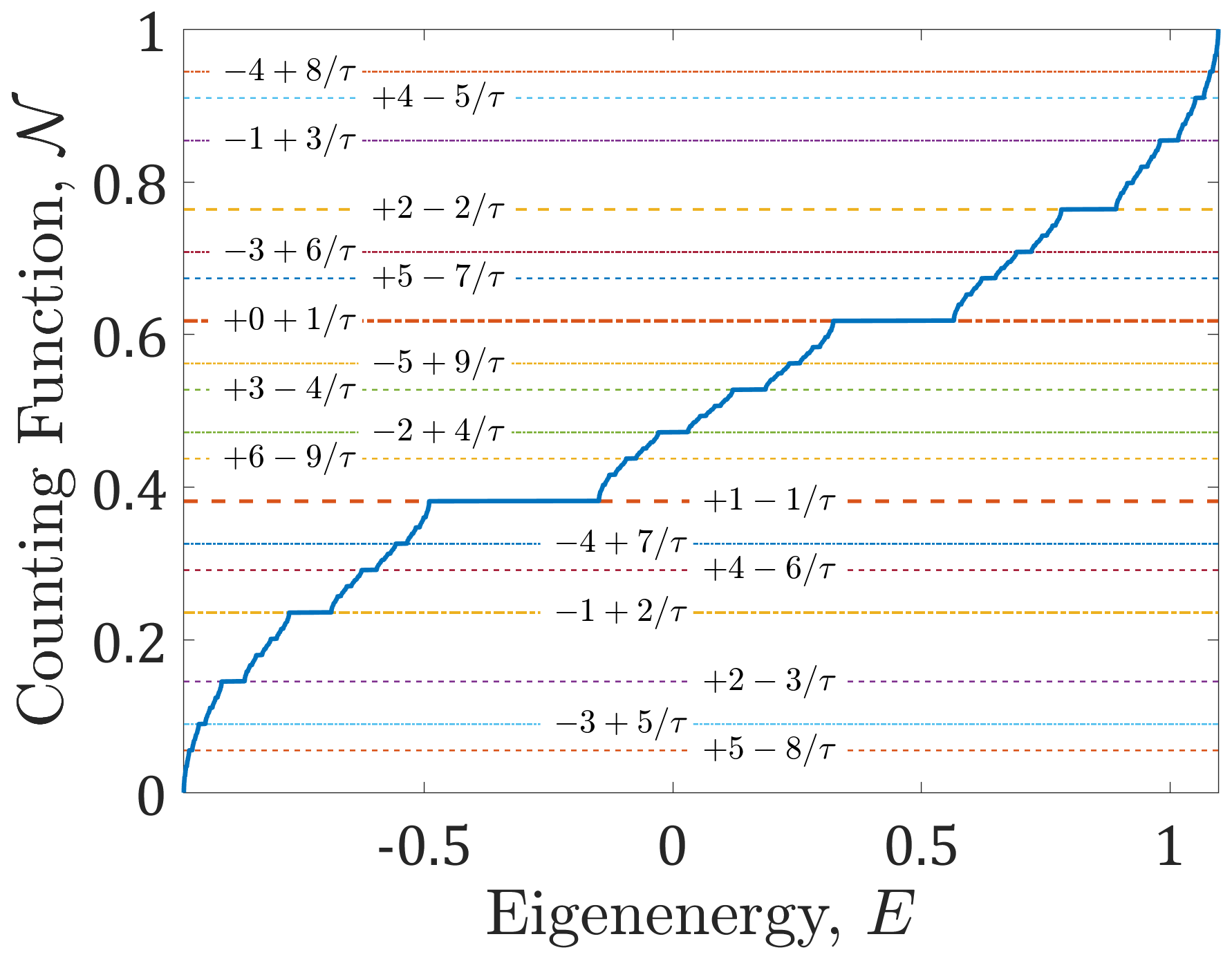} }
\hfill{}
\subfloat[\label{fig:Thue-Morse_count}Thue-Morse.]
{\includegraphics[width=0.3\textwidth]{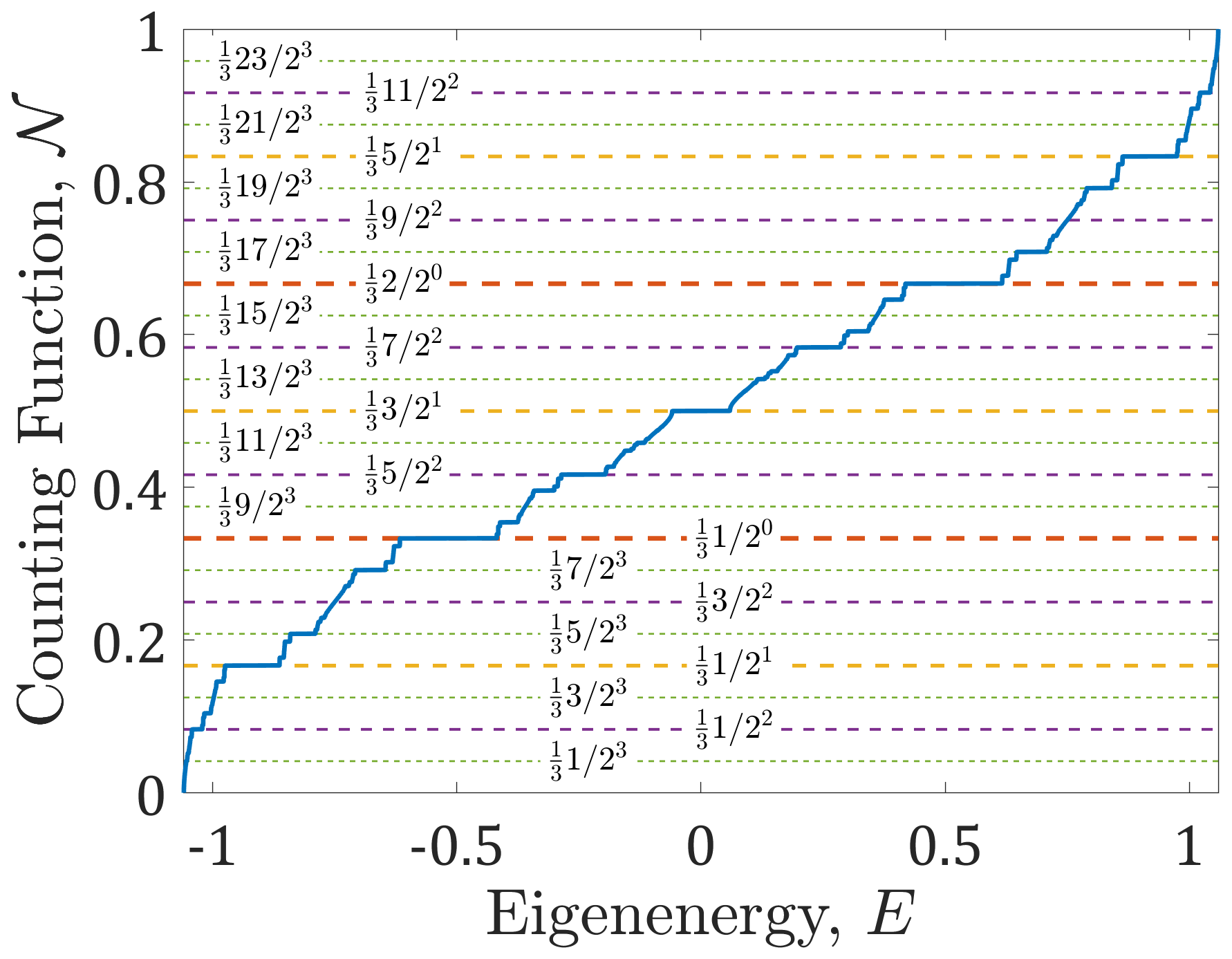} }
\\
\hfill{}
\hfill{}
\subfloat[\label{fig:Period-Doubling_count}Period Doubl\rlap{ing.}]
{\includegraphics[width=0.3\textwidth]{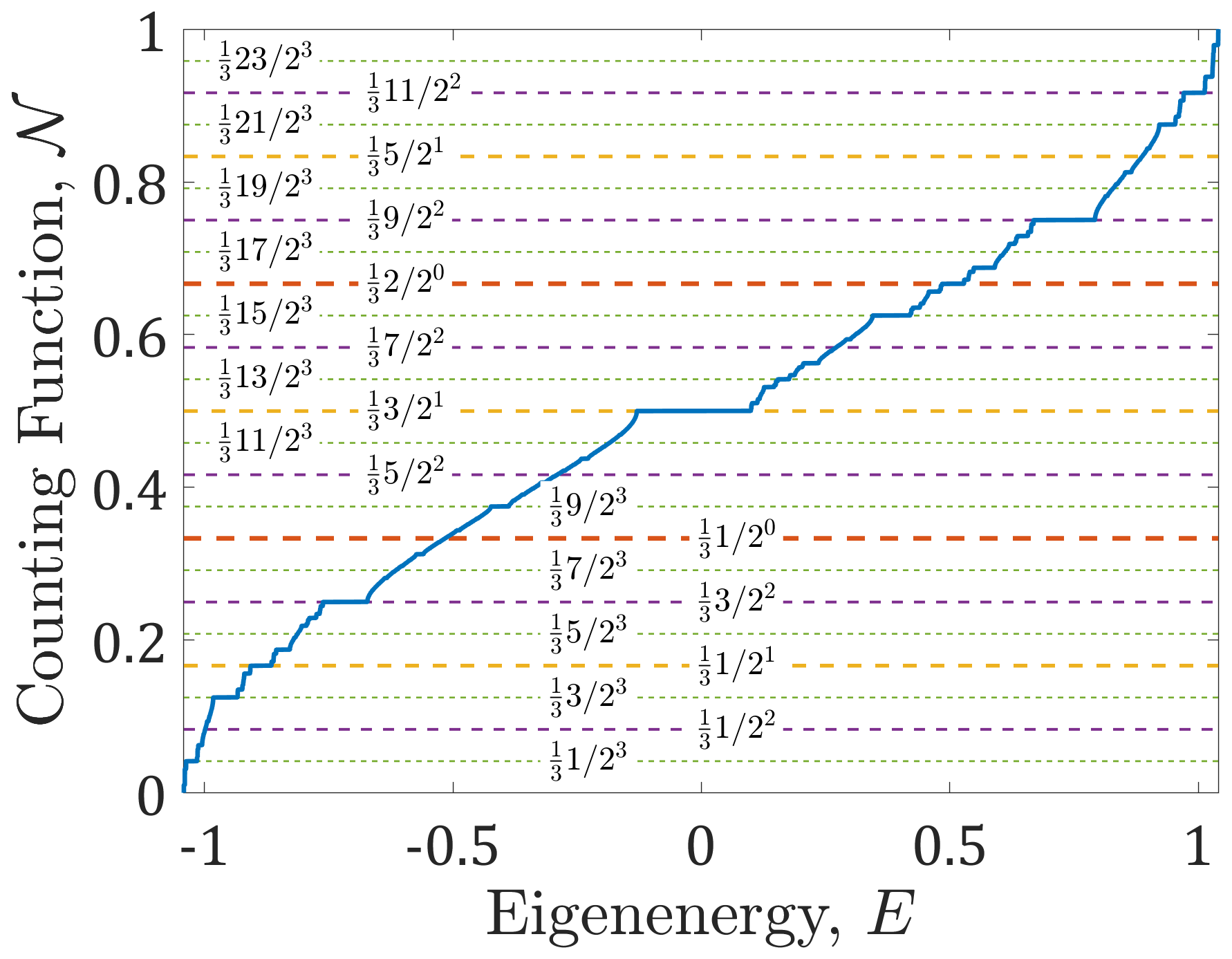} }
\hfill{}
\subfloat[\label{fig:Rudin-Shapiro_count}Rudin-Shapiro.]
{\includegraphics[width=0.3\textwidth]{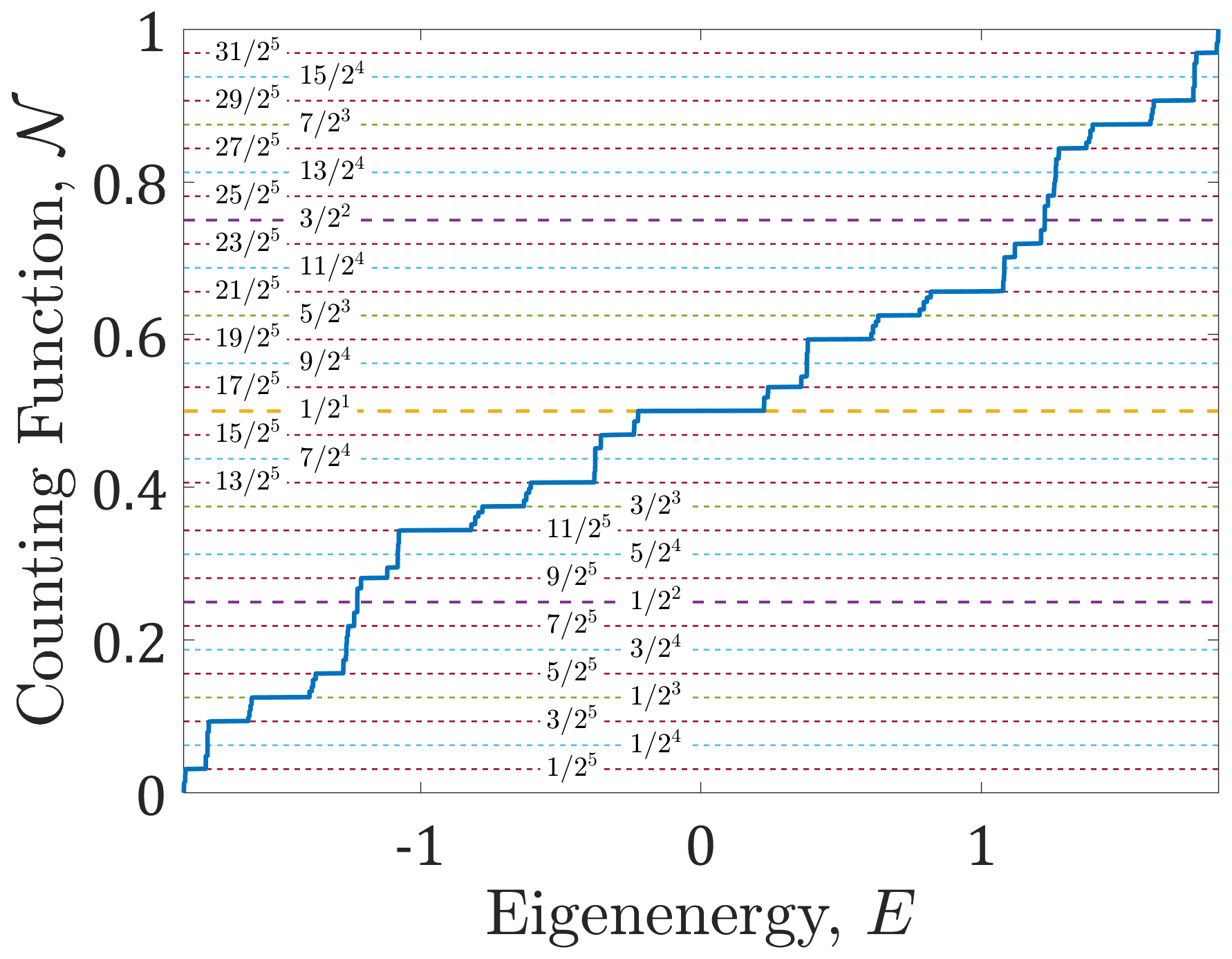} }
\hfill{}
\hfill{}
\caption{\label{fig:count}
    Counting Functions of representatives of five families of tilings whose diffraction spectra are displayed in Fig.~\ref{fig:diff}. (\textsc{a}) describes a periodic tiling with two gaps corresponding to the Bragg peaks in the Brillouin zone, (\textsc{b}) a quasiperiodic tiling (Fibonacci), (\textsc{c}) a non-quasiperiodic tiling (Thue-Morse), (\textsc{d}) a limit quasiperiodic tiling (Period Doubling) and (\textsc{e}) a non-quasiperiodic and non-Pisot tiling (Rudin-Shapiro). In contrast to the periodic case, note that for the other examples of aperiodic tilings, there is an infinite number of spectral gaps. This fractal structure (discrete scaling symmetry) is typical of aperiodic tilings.
}
\end{figure}

For periodic atomic arrangements, the Bloch theorem indicates that  the spectrum of \eqref{tb5} consists of bands and hence gaps whose locations are directly related to the (pure-point) Bragg diffraction spectrum as displayed in Fig.~\ref{fig:periodic_Bloch}. 

For aperiodic tilings, the gap labeling theorem (GLT) is an important  and elegant result valid in  space dimensions $d \leq 3$, that allows to calculate systematically the counting function at gap values  \citep{Bellissard_RMP_1991,Bellissard_Review_1992,BBG,Kellendonk_1995,Kellendonk_1997,Sadun}. The GLT states that possible values of $\Nc(e)$ in the gaps are given by all possible letter frequencies of all possible words in the infinite tiling generated by a substitution. Those frequencies can be expressed as linear combinations of the frequencies of one and two letters words only, namely using the left eigenvector $\mathbf{v}_1 = (\rho_a, \rho_b)$ previously defined in \eqref{eq:rho_b}.
The GLT makes use of $K$-theory, identifies the $K_0$ group and it allows to systematically build the gap labeling group as a trace,
$\tk\co K_0 ( \zone ) \to \RR $. 
For quasicrystals, the gap labeling group is 
\begin{equation}
\tk ( K_0 (\zone)) = ( \ZZ + \rho_b \, \ZZ ) \cap [0,1] \, ,
\label{glt}
\end{equation}
so that every gap is labelled by two integers. 

\medskip
A systematic calculation \citep{Bellissard_Review_1992} gives (see Table~\ref{tab:examples} for more examples):
\medskip
 
\begin{Pro} ~
\begin{enumerate}
\item For the quasiperiodic Fibonacci tiling,
\[
\tk (K_0(\zone)) = (\ZZ + \lambda ^{-1}\ZZ) \cap [0,1] \, ,
\]
using $\rho_b=1-\lambda^{-1}$ in \eqref{eq:rho_b}.
\medskip
\item 
For the aperiodic but non-quasiperiodic Thue-Morse tiling,  
\[
\tk (K_0(\zone)) = 
\big(\tfrac 13\,\ZZ \big[\tfrac 12\big]  \big) \cap [0,1] \, .  
\]
\end{enumerate}
\end{Pro} 
\medskip


\begin{figure}[tb]
\centering
\subfloat[\label{fig:periodic_Bloch}Periodic.]
{\includegraphics[width=0.3\textwidth]{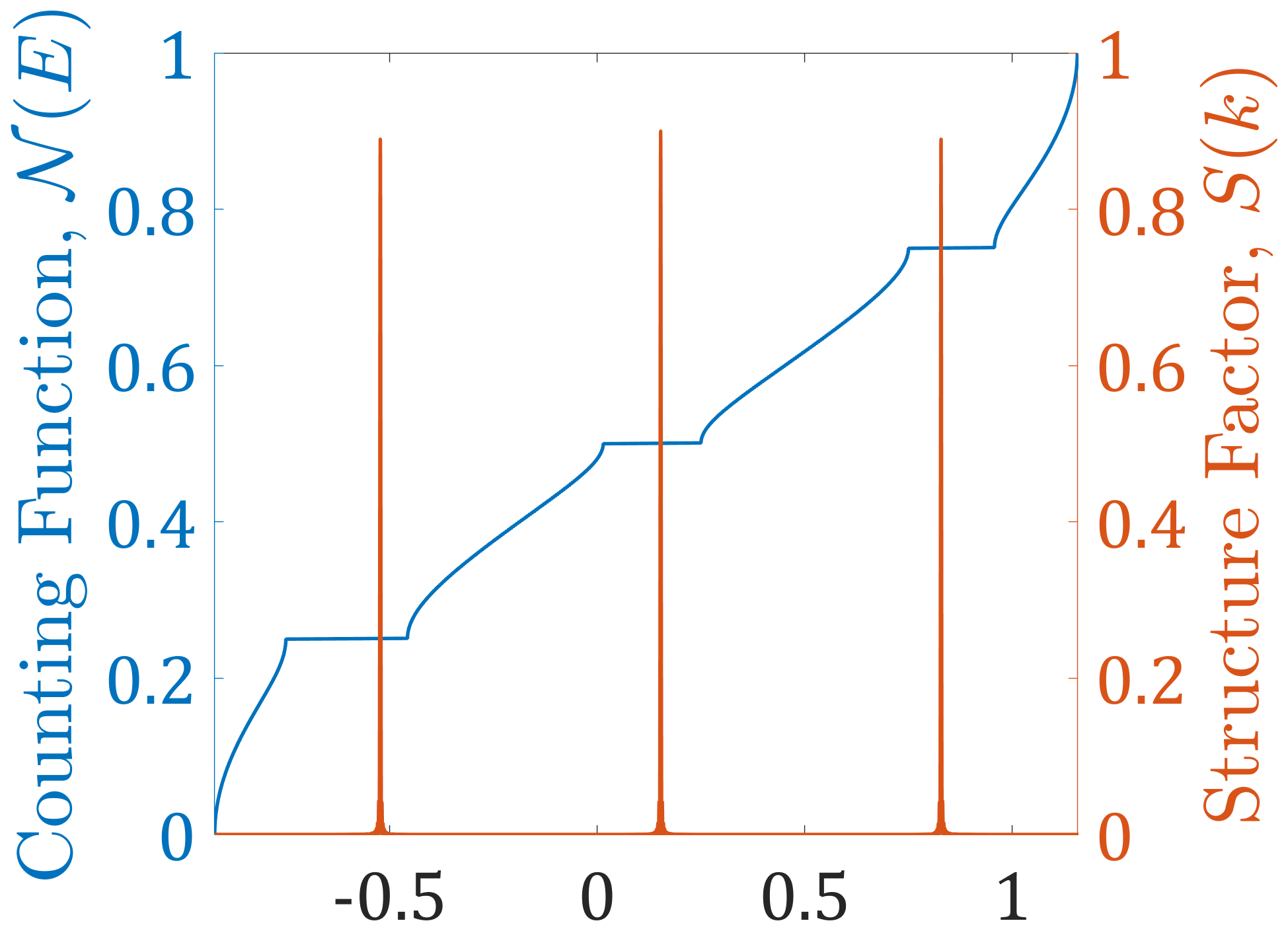} }
\hfill{}
\subfloat[\label{fig:Fibonacci_Bloch}Fibonacci.]
{\includegraphics[width=0.3\textwidth]{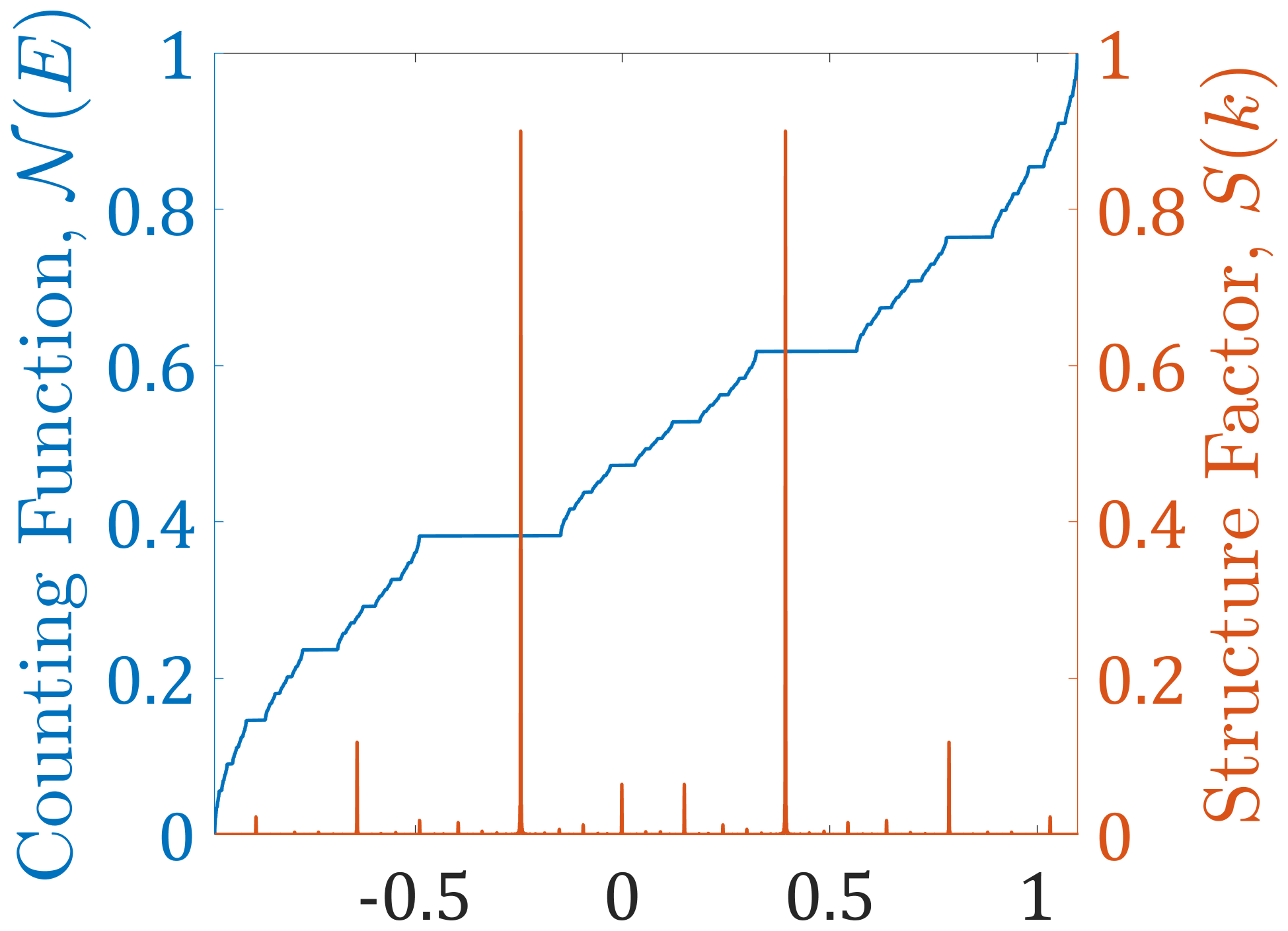} }
\hfill{}
\subfloat[\label{fig:Thue-Morse_Bloch}Thue-Morse.]
{\includegraphics[width=0.3\textwidth]{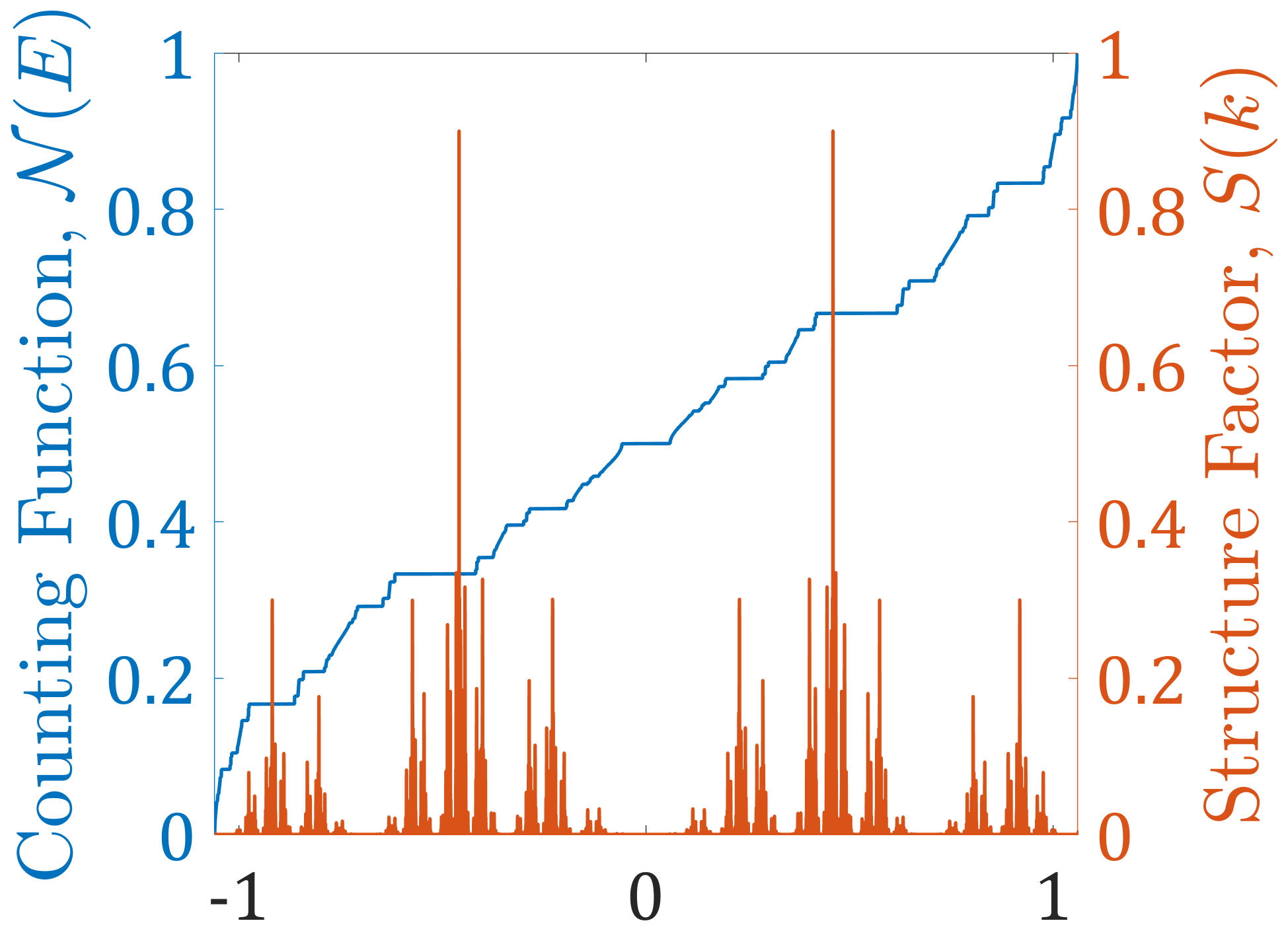} }
\\
\hfill{}
\hfill{}
\subfloat[\label{fig:Period-Doubling_Bloch}Period Doubl\rlap{ing.}]
{\includegraphics[width=0.3\textwidth]{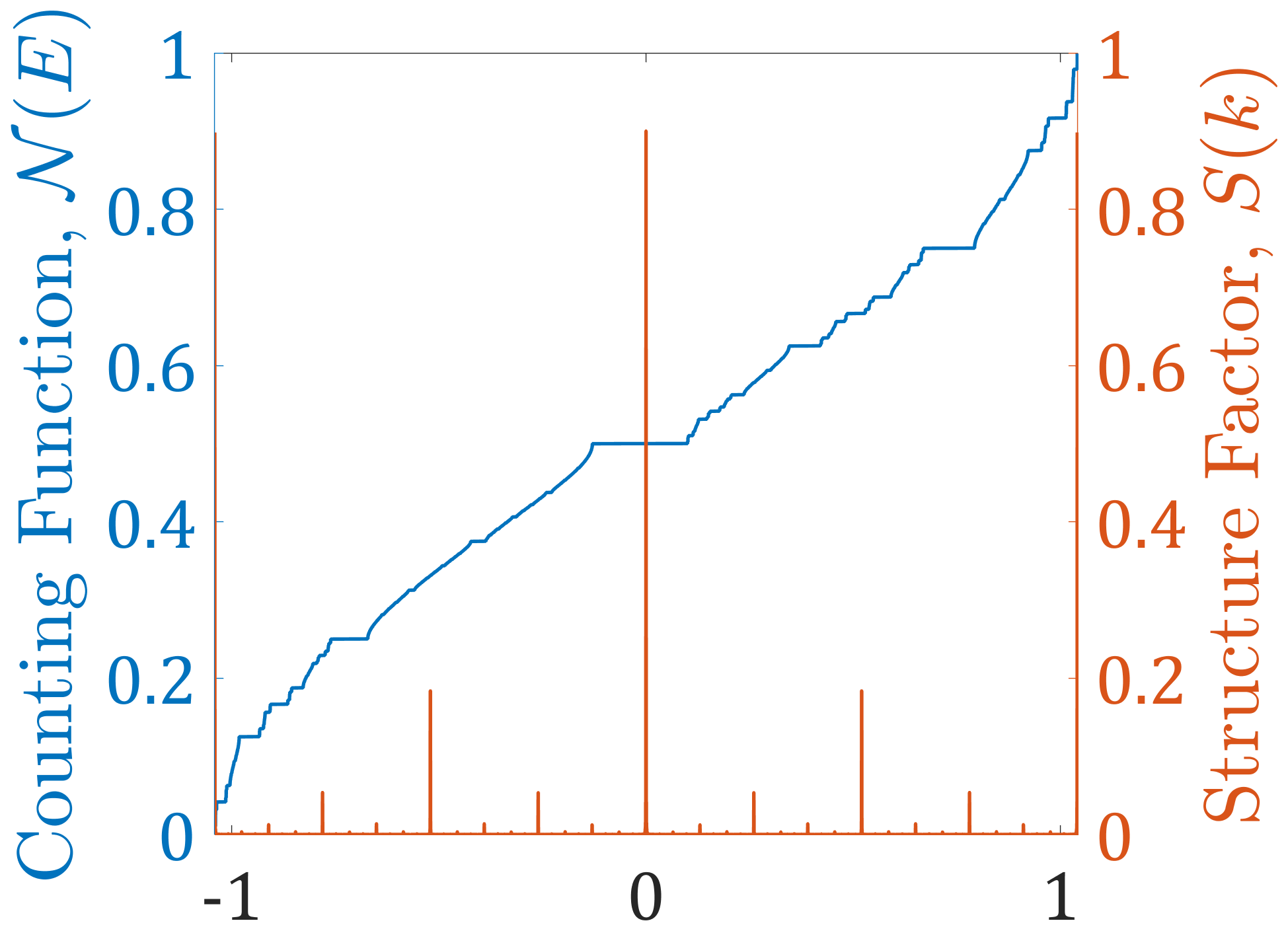} }
\hfill{}
\subfloat[\label{fig:Rudin-Shapiro_Bloch}Rudin-Shapiro.]
{\includegraphics[width=0.3\textwidth]{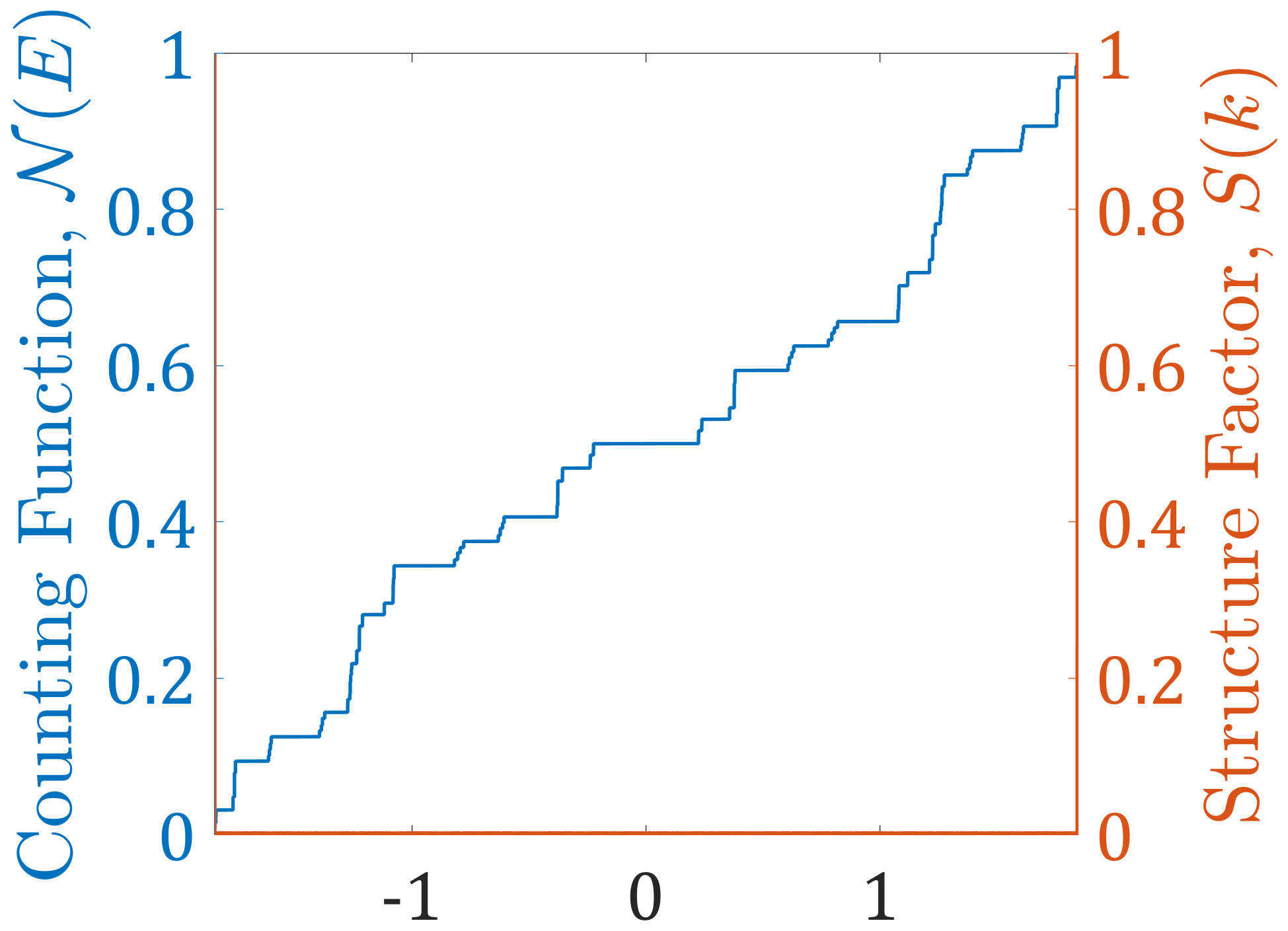} }
\hfill{}
\hfill{}
\caption{\label{fig:Bloch}
    Comparison between diffraction and spectral data for the five representative families of one-dimensional tilings considered previously and to which Theorem~\ref{thm:maintiling} applies. For the periodic (\textsc{a}), quasiperiodic (\textsc{b}) and aperiodic (limit-quasiperiodic) (\textsc{d}) tilings, there is a direct correspondence between the two sets of data. This can be viewed as an extension of the Bloch theorem. Note that for these three cases, the diffraction spectrum is PP, a result to be contrasted with the non-quasiperiodic Pisot Thue-Morse (\textsc{c}) and the aperiodic Rudin-Shapiro (\textsc{e}) tilings for which the diffraction spectrum is respectively SC and AC, while the spectral counting function accounts for infinitely countable gaps well described by the GLT. 
}
\end{figure}

Comparing \eqref{cech} and \eqref{glt}, it cannot escape our attention that for the Fibonacci tiling, Bragg peaks and spectral gaps locations are in one-to-one correspondence, a result strongly reminiscent of the Bloch theorem for periodic structures. Its extension to classes of aperiodic tilings, asserting the commutativity of the diagram
\begin{equation}
\xymatrix{
K_0(\zone _T)  \ar@{>>}[rr]^(.45){(\ch _d) \circ (\chi ^{-1})}  \ar[d]^{\tk}
&&\CH^d(\hull ; \ZZ ) \ar[d]^{\tc} \\
\RR \ar[rr]^{\id} &&\RR}
\end{equation}
with $(\ch _d) \circ (\chi ^{-1})$ onto 
provides an explanation for this empirical result.  
Applying this general result to the Fibonacci example yields 
the commuting diagram
\begin{equation}
\xymatrix{\ZZ \oplus \ZZ\ar[r]^{\id}\ar[d]^{\tk}&\ZZ \oplus \ZZ\ar[d]^{\tc} \\
\RR\ar[r]^{\id}&\RR}
\end{equation}
with
\begin{equation}
\xymatrix{(p,q)\ar@{|->}[r]^{\id}\ar@{|->}[d]^{\tk}& (p,q)\ar@{|->}[d]^{\tc} \\
p + \rho_b\,q \ar@{|->}[r]^{\id}& p + \rho_b\,q }
\end{equation}
which can indeed be viewed as a generalisation of Bloch theorem to quasicrystals. At a general level, it is not too surprising a result, since we have observed that structural and spectral data of quasiperiodic substitutions have been deduced from the appearance frequencies of the single and double letters tiles by means of the left eigenvector  $\mathbf{v}_1 = (\rho_a , \rho_b )$. For the $K_{0}(\zone)$ group, it is by construction. On the other hand, the \v Cech cohomology $\CH^{1}\left(\Omega_{T};\ZZ\right)$ contains additional information about the order of the tiles. But for quasiperiodic substitutions, the order of the tiles is irrelevant, thus leading to the equivalence between the two groups. In more complicated cases, e.g., the Thue-Morse tiling, the order of the tiles plays a role. It is easy to see this from the corresponding substitution matrix $ M= 
\begin{psmallmatrix}1 & 1\\
1 & 1
\end{psmallmatrix}$ which is identical to the periodic substitution $a\,\mapsto\,ab $ and $
b\,\mapsto\,ab$ but with very different structural (order) and spectral (two letter frequencies) properties.

In order to further clarify the content of Theorem \ref{thm:maintiling} and its conditions of applicability, we wish to first discuss features of the Thue-Morse aperiodic tiling. 
It is not a C\&P quasicrystal, yet it is a Pisot substitution. The possible values of $\Nc(e)$ in the gaps are obtained from the gap labeling group $\tk ( K_0 ( \zone ))$. This is to be compared to the diffraction spectrum composed of Bragg peaks (PP) and of a SC broad range contribution  which does not appear in the cohomology trace $\tc (\CH ^1(\Omega_T ; \ZZ )) = \tfrac 13 \, \ZZ\big[\tfrac 12\big]$. This lack of equivalence between diffraction and (Laplacian) spectral data  is not a limitation of Theorem \ref{thm:maintiling} since the Thue-Morse tiling abides its conditions of applicability. It is the expression of a discrepancy between $\tc (\CH ^1(\Omega_T ; \ZZ ))$ and the structure factor $S(k)$ which contains additional information not accessible from the cohomology description. It is interesting to note though that in a detailed experimental measurement of the Thue-Morse diffraction spectrum \cite{Axel_1991}, it was effectively challenging to observe diffraction peaks other than those predicted by the cohomology trace $\tfrac 13 \, \ZZ\big[\tfrac 12\big]$. Furthermore, the aforementioned lack of equivalence is unrelated to the lack of periodicity or quasiperiodicity (e.g.\ Period Doubling tiling) but rather a consequence of the nature of the diffraction spectrum, a quantity which, unlike spectral data, is sensitive to both local symmetries of the tiles, a condition of applicability of Theorem \ref{thm:maintiling}, and to long-range correlations driven e.g.\ by the order of the letters (immaterial for periodic or C\&P quasicrystals). For instance, the Rudin-Shapiro tiling has an absolutely continuous and structureless diffraction spectrum but a fractal spectral gap distribution well accounted by $\tc\left(\CH^1\left(\Omega _T;\mathbb{Z}\right)\right) = \ZZ \bigl[\frac12\bigr]$.  
These features are summarised in Table~\ref{tab:examples}.


\begin{sidewaystable}
\centering{}
\caption{\label{tab:examples}
   Summary of our results applied to main representatives of $1d$ tilings. For each of them, we have indicated the \v Cech cohomology $\CH^{1}\left(\Omega_{T};\ZZ\right)$, the nature of the diffraction spectrum, pure-point (PP), absolutely continuous (AC) and singular continuous (SC). Theorem \ref{thm:maintiling} applies to all cases so that the cohomology trace $\tc\left(\CH^1\left(\Omega_{T};\ZZ\right)\right)$ is calculated using the trace of the $K_{0}(\zone)$ group. Here, $\lambda =(\sqrt{5}+1)/2$ is the golden ratio with $\rho_b = 1-\lambda^{-1}$, and $n,p,q,m,N\in\ZZ$.
}
\smallskip
\begin{tabular}{M{10em}lE{3.5em}F{3.5em}lcE{2.5em}F{2.5em}}
\toprule 
Family  & $\CH^{1}\left(\Omega _T;\mathbb{Z}\right)$  & \multicolumn{3}{m{10em}}{Diffraction peaks} & $\tc\left(\CH^1\left(\Omega _T;\mathbb{Z}\right)\right)$ & \multicolumn{2}{m{6em}}{Spectral gaps (modulo 1)} \tabularnewline
\addlinespace
\midrule 
Periodic  & $\ZZ^{\hphantom{1}}$  & $k_{n}$&$n$ & PP & $\ZZ$ & $\Nc$&$\text{const}$ \tabularnewline
Fibonacci  & $\ZZ^{2}$  & $k_{p,q}$&$p+q/\lambda$ & PP & $\ZZ +\lambda^{-1}\ZZ$  & $\Nc_{q}$&$q/\lambda$ \tabularnewline
Thue-Morse  & $\ZZ\oplus\ZZ\bigl[\frac{1}{2}\bigr]$  & $k_{n,m,N}$&$\frac{1}{2n+1}\,\frac{m}{2^{N}}$  & SC+PP & $\frac{1}{3}\,\ZZ\bigl[\frac{1}{2}\bigr]$ & $\Nc_{m,N}$&$\frac{1}{3}\,\frac{m}{2^{N}}$ \tabularnewline
Period Doubling & $\ZZ\oplus\ZZ\bigl[\frac12\bigr]$ & $k_{m,N}$&$\frac{m}{2^N}$ & PP & $\frac13\, \ZZ \bigl[\frac12\bigr]$ & $\Nc_{m,N}$&$\frac13\,\frac{m}{2^N}$ \tabularnewline
Rudin-Shapiro & $\ZZ\oplus\ZZ\bigl[\frac12\bigr]\oplus\ZZ^2\bigl[\frac12\bigr]$ & \multicolumn{2}{c}{N/A} & AC & $\hphantom{\frac11\,}\ZZ \bigl[\frac12\bigr]$ & $\Nc_{m,N}$&$\frac{m}{2^N}$ \tabularnewline
\addlinespace
\bottomrule
\end{tabular}
\end{sidewaystable}


\section{Insights into prior work}

This section is designed to link up the discussion of traces in the
previous sections with the so-called Shubin trace used in the 
mathematical-physics literature.   We discuss the papers of Shubin \cite{shu},
Lenz, Peyerimhoff, and Veseli\'c \cite{LPV},  Moustafa \cite{M}, 
Kriesel \cite{Kr}, and Benameur-Mathai \cite{BM}.

\medskip

\begin{description}

\item[] \emptylabel
The equivalence found for  quasicrystals between the structural Fourier module and the counting function 
\begin{equation}
C_\nu \left( \CH^{1}\left(\Omega_{T};\ZZ\right) \right) = \tk \left( K_0 (\Omega_T ;\ZZ )\right)= \ZZ + \rho_b \, \ZZ \, ,
\label{equivalence}
\end{equation} 
respectively described by the two traces $C_\nu \left( \CH^{1}\left(\Omega_{T};\ZZ\right) \right)$ and $\tk \left( K_0 (\Omega_T ;\ZZ )\right)$ has been first noticed in \textbf{R. Johnson and J. Moser} \citep{Johnson_1982,Johnson_1986} (1982). These works  studied  the spectrum of self-adjoint linear differential operators and presented a systematic way of enumerating the open intervals of the associated resolvent operator (GLT), using the rotation number. This was an alternative treatment of gap labeling also in the spirit of the Schwartzman winding number \cite{Schwartzman} and using cohomology ideas. A discrete version is in \citep{delyon1983rotation}. Yet, this interpretation was based on the use of the rotation number, a quantity neither related to the \v Cech cohomology nor to 
$C_\nu \left( \CH^{1}\left(\Omega_{T};\ZZ\right) \right)$. 
Moreover, it is not obviously generalizable to higher dimensions. 
Note that formula \eqref{equivalence} is just a special case of our
Corollary \ref{cor:maintiling}.

\item[Shubin's] paper \cite[formula (2.3)]{shu} (1994) is all about the
irrational rotation $C^*$-algebra $A_\alpha  $. For our purposes, we
consider the Kronecker flow on the torus. Let $N$ be a circle that is
transverse to the foliation with its natural $\ZZ$-action given by the
foliation.  Its natural Lebesgue measure 
is an invariant transverse measure $\nu $  for the foliation. Then the
$C^*$-algebra $C(N) \rtimes \ZZ $ sits in the von
Neumann algebra $L^\infty (N) \rtimes \ZZ $, which is a $I\!I_1$ factor
with trace associated to $\nu$.  It is denoted $W_\alpha $
by Shubin. Its trace, given as (2.8) in Shubin's paper, is exactly the
discrete version of the canonical trace on $A_\alpha $, which sits in the 
$I\!I_\infty$  factor  $W_\alpha \otimes \mathcal{B}(\mathcal {H})$.
The situation is very much like that for tiling spaces, except
that in this example $N$ is a circle instead of being $0$-dimensional.

\item[J. Bellissard,  R. Benedetti, and J.-M.  Gambaudo] \cite{BBG} (2006)
defined the trace that they use initially by taking advantage of the
fact that the groupoid $C^*$-algebra is a crossed product. 
Let $A$ be the dense subalgebra of $C(\Omega _T) \rtimes \RR ^d $
consisting of continuous compactly supported 
functions $\Omega _T \times \RR ^d \to \CC $ and let $\mu $ be an invariant
probability measure on $\Omega _T$.  The   trace 
$\tau _\mu $  is defined by 
\[
\tau _\mu (f) = \int _{x \in \Omega _T} f(x, 0)\, d\mu (x)  .     
\] 
This is exactly as we have described the $K$-theory trace.

\item[] \emptylabel
To see the connection with the paper of \textbf{Lenz, Peyerimhoff, and Veseli\'c}
\cite{LPV} (2007) we have to do some translation.  They use the terminology 
of A.\ Connes in his original treatment of the foliation index theorem.
We have been using instead the Moore-Schochet terminology, which we 
prefer. In \cite[Section 4]{LPV} the authors define the canonical
trace on the von Neumann algebra of the foliation which we denote
$W^*(G(\Omega _T), \tilde \mu))$. Their Theorem 4.2
demonstrates the use of locally traceable operators 
 in the construction of the $K$-theory trace.

\item[Andress and Robinson] \citep{AR} (2010) 
explicitly use cohomology to study tilings.  We convert to 
our terminology to state their results.
Let $N$ be a Cantor set and $T$ a homeomorphism which is strictly
ergodic (minimal and uniquely ergodic.) Let $\mu $ denote the 
invariant measure. Then we may form the suspension
$\hull = N \times _{\ZZ} \RR $ with associated measure $\mu '$ on $\hull$.
Andress and Robinson define the first \v Cech group  $\CH^1(\hull; \ZZ ) $, the
coinvariant group $\CH^0(N ; \ZZ )_{\ZZ}  $ (which they call the
\emph {dynamic cohomology} group), and prove that 
\[
\CH^0(N ; \ZZ )_{\ZZ}  \cong \CH^1(\hull; \ZZ ),
\]
anticipating  Theorem \ref{coinvariant}.  

Next, they recall the Schwartzman winding number \citep{Schwartzman}, which   is a 
real-valued functional on  
\[
\CH^1(\hull; \ZZ ) \cong [\hull, S^1] 
\]
defined on continuous functions which are continuously 
differentiable along the leaves (which are dense)
$f: \hull \to S^1$ by
\[
W(f) = \frac{1}{2\pi i} \int_{\hull} \frac{f'(y)}{f(y)} \, d\mu'(y) .
\]
(One could just as well work with continuous functions smooth along the
leaves.)
Let $\WW (\hull)$ denote the set of values of the Schwartzman function. This
is a countable subgroup of $\RR $.  Regarding elements of
$\CH^1(\hull; \ZZ ) $ as eigenfunctions, 
then $\WW (\hull)$ is the set of all eigenvalues associated with the
tiling.  Now it is possible that $\ker(W) \neq 0$;  such classes are
called in \cite{AR} \emph{cohomologically invariant}. They 
are represented by continuous functions $f\co\hull \to S^1 $ with
$\int d\log f d\mu' = 0$.  Coboundaries have this property but possibly
other functions do as well. The usual situation is for $\ker(W) = 0$
on $\CH^1(\hull; \ZZ)$, in which case the tiling is said
to be \emph{saturated}.  Of course this is equivalent to the statement that
the winding map  $W\co \CH^1(\hull; \ZZ) \longrightarrow \WW(\hull)$
is an isomorphism. When this fails there are ``invisible''
{\Cech} cohomology classes not detected by the trace.

To further fit this work into our framework, we note that the unique
ergodicity assumption forces 
\[
\bar H_\tau^1(\hull) \cong \RR
\] 
(via the unique trace) and the canonical map
$s\co\CH^1(\hull; \ZZ) \longrightarrow  \bar H_\tau^1(\hull) \cong \RR$
coincides with the Schwartzman winding number, essentially by uniqueness 
of the ergodicity.

\item[\textmd{The connection with the paper of} Moustafa] 
\cite{M} (2010) is the
easiest to make, since he explicitly  uses the partial 
Chern character (which he calls the
\emph{longitudinal Chern character}) and tangential cohomology
(\emph{longitudinal  cohomology}).

\item[Kriesel]  \cite{Kr} (2016) makes use of Gabor frames to generalize
the gap labeling theorem to the situation of non-trivial magnetic
fields. His treatment is entirely $K$-theoretic, and includes
precise results on gap labelling in dimension $2$.

\item[Benameur and Mathai] \cite{BM} (2020) further pursue the topic of
$K$-theoretic traces in the presence of a magnetic field.
They give \cite[Appendix B]{BM} a history of gap-labelling theorems
and present two gap-labelling conjectures, which to our knowledge
are still open in general, though they prove a number of special cases.

\end{description}


\providecommand{\bysame}{\leavevmode\hbox to3em{\hrulefill}\thinspace}
\providecommand{\MR}{\relax\ifhmode\unskip\space\fi MR }
\providecommand{\MRhref}[2]{%
  \href{http://www.ams.org/mathscinet-getitem?mr=#1}{#2}
}
\providecommand{\href}[2]{#2}

\end{document}